\documentclass[english]{article}

\usepackage{amsmath}
\usepackage{amssymb}

\usepackage{cite}

\usepackage{subfig,floatrow}

\usepackage{color}
\usepackage{tikz}
\usetikzlibrary{positioning}
\usetikzlibrary{shapes}
\usetikzlibrary{patterns}
\usetikzlibrary{arrows}
\tikzset{every picture/.style={>=stealth'}}

\usepackage{graphicx}
\usepackage{pgfplots}
\pgfplotsset{compat=1.16}
\usepgfplotslibrary{fillbetween}

\makeatletter
\newenvironment{customlegend}[1][]{%
    \begingroup
    \pgfplots@init@cleared@structures
    \pgfplotsset{#1}%
}{%
    \pgfplots@createlegend
    \endgroup
}%

\def\addlegendimage{\pgfplots@addlegendimage}
\makeatother

\usepackage{listings}
\lstdefinelanguage{pseudo}{
  morekeywords=[1]{
    break, continue, do, each, else, for, if, loop, let, matches, od, otherwise,
    repeat, return, then, until, while,
    Apply, MergeMaxUid, PeekMax, PopMax, Reduce, RequestsFor, ShortCircuitingRequestsFor, TopOf,
    },
  morecomment=[l]{//},
  morecomment=[s]{/*}{*/},
  numbers=left,
  literate={:=}{{$\gets$}}1
  {<=}{{$\leq$}}1
  {>=}{{$\geq$}}1
  {<>}{{$\neq$}}1
  {!}{{$\neg$}}1
  {->}{{$\rightarrow$}}1
  {=>}{{$\Rightarrow$}}1
}

\usepackage{hyperref,doi}
\usepackage{multirow}

\usepackage{graphicx}

\newcommand{\B}[0]{\ensuremath{\mathbb{B}}}
\newcommand{\N}[0]{\ensuremath{\mathbb{N}}}

\newcommand{\sort}[0]{\text{sort}}

\newcommand{\Adiar}[0]{Adiar}
\newcommand{\TPIE}[0]{TPIE}

\newcommand{\QappA}[0]{\ensuremath{Q_{\mathit{app}:1}}}
\newcommand{\QappB}[0]{\ensuremath{Q_{\mathit{app}:2}}}
\newcommand{\Lred}[0]{\ensuremath{F_{\mathit{internal}}}}
\newcommand{\Lredleaf}[0]{\ensuremath{F_{\mathit{leaf}}}}
\newcommand{\Qred}[0]{\ensuremath{Q_{\mathit{red}}}}

\newcommand{\Lout}[0]{\ensuremath{F_{\mathit{out}}}}

\newcommand{\Lj}[0]{\ensuremath{F_{j}}}

\newcommand{\LjredA}[0]{\ensuremath{F_{j:1}}}
\newcommand{\LjredB}[0]{\ensuremath{F_{j:2}}}

\newcommand{\ehigh}[0]{\ensuremath{e_{\mathit{high}}}}
\newcommand{\elow}[0]{\ensuremath{e_{\mathit{low}}}}

\newcommand{\triple}[3]{\ensuremath{(#1, #2, #3)}}
\newcommand{\arc}[3]{\ensuremath{#1 \xrightarrow{_{#2}} #3}}
\newcommand{\carc}[4]{\ensuremath{#1\ \textcolor{#4}{\xrightarrow{\textcolor{black}{_{#2}}}}\ #3}}

\def\orcidID#1{\smash{\href{http://orcid.org/#1}{\protect\raisebox{-1.25pt}{\protect\includegraphics{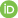}}}}}

\hypersetup{
    colorlinks,
    linkcolor={red!50!black},
    citecolor={green!40!black},
    urlcolor={blue!40!black}
}

\usepackage{amsthm}

\newtheorem{theorem}{Theorem}[section]
\newtheorem{lemma}[theorem]{Lemma}
\newtheorem{proposition}[theorem]{Proposition}
\newtheorem{corollary}[theorem]{Corollary}

\tikzstyle{plot_adiar}=[color=black, mark=o, mark size=1pt, line width=0.7pt]
\tikzstyle{plot_buddy}=[color=red, mark=triangle, mark size=1pt, line width=0.7pt]
\tikzstyle{plot_cudd}=[color=blue, mark=diamond, mark size=1pt, line width=0.7pt]
\tikzstyle{plot_sylvan}=[color=purple, mark=square, mark size=1pt, line width=0.7pt]

\def\arxiv{1}

\begin{document}

\title{Efficient Binary Decision Diagram Manipulation in External
  Memory\footnote{This is the full version of the TACAS 2022 paper \cite{Soelvsten2022:TACAS}}}

\author{Steffan Christ Sølvsten
    \\ \\ {\normalsize Aarhus University, Denmark}
    \\ \\ Anna Blume Jakobsen
    \and Jaco van de Pol
    \\ \\ {\normalsize \{%
      \href{mailto:soelvsten@cs.au.dk}{\color{black} soelvsten},%
      \href{mailto:jaco@cs.au.dk}{\color{black} jaco}%
      \}@cs.au.dk
      }
    \\ \\
    Mathias Weller Berg Thomasen
  }

% typeset the header of the contribution
\maketitle

% ---------------------------------------------------------------------------- %

\begin{abstract}
  We follow up on the idea of Lars Arge to rephrase the Reduce and Apply
  procedures of Binary Decision Diagrams (BDDs) as iterative I/O-efficient
  algorithms. We identify multiple avenues to simplify and improve the
  performance of his proposed algorithms. Furthermore, we extend the technique
  to other common BDD operations, many of which are not derivable using Apply
  operations alone, and we provide asymptotic improvements for the procedures
  that can be derived using Apply.

  These algorithms are implemented in a new BDD package, named \Adiar. We see
  very promising results when comparing the performance of \Adiar\ with
  conventional BDD packages that use recursive depth-first algorithms. For
  instances larger than $8.2$~GiB\footnote{An error in the data analysis of
    resulted in this threshold being reported as $9.5$~GiB in
    \cite{Soelvsten2022:TACAS}. This larger number was the maximum size of the
    entire BDD forest rather than only the largest single BDD as was intended.},
  our algorithms, in parts using the disk, are $1.47$ to $3.69$ times slower
  compared to CUDD and Sylvan, exclusively using main memory. Yet, our proposed
  techniques are able to obtain this performance at a fraction of the main
  memory needed by conventional BDD packages to function. Furthermore, with
  \Adiar\ we are able to manipulate BDDs that outgrow main memory and so surpass
  the limits of other BDD packages.
\end{abstract}

% ---------------------------------------------------------------------------- %

\section{Introduction} \label{sec:introduction}

A Binary Decision Diagram (BDD) provides a canonical and concise representation
of a boolean function as an acyclic rooted graph. This turns manipulation of
boolean functions into manipulation of directed acyclic graphs
\cite{Bryant1986,Bryant1992}.

Their ability to compress the representation of a boolean function has made them
widely used within the field of verification. BDDs have especially found use in
model checking, since they can efficiently represent both the set of states and
the state-transition function \cite{Bryant1992}. Examples are the symbolic model
checkers NuSMV \cite{Cimatti2000,Cimatti2002}, MCK \cite{Gammie2004},
\textsc{LTSmin} \cite{Kant2015}, and MCMAS \cite{Lomuscio2017} and the recently
envisioned symbolic model checking algorithms for CTL* in \cite{Amparore2020}
and for CTLK in \cite{He2020}. Bryant and
Marijn~\cite{Bryant2020,Bryant2021:1,Bryant2021:2} also recently devised how to
use BDDs to construct extended resolution proofs to verify the result of SAT and
QBF-solvers. Hence, continuous research effort is devoted to improve the
performance of this data structure. For example, despite the fact that BDDs were
initially envisioned back in $1986$, BDD manipulation was first parallelised in
$2014$ by Velev and Gao~\cite{Velev2014} for the GPU and in $2016$ by Dijk and
Van de Pol~\cite{Dijk2016} for multi-core processors \cite{Bryant2018}.

The most widely used implementations of decision diagrams make use of recursive
depth-first algorithms and a unique node table
\cite{Karplus1988,Brace1990,Somenzi2015,Lind1999,Dijk2016}. Lookup of nodes in
this table and following pointers in the data structure during recursion both
pause the entire computation while missing data is fetched
\cite{Klarlund1996,Minato2001}. For large enough instances, data has to reside
on disk and the resulting I/O-operations that ensue become the bottle-neck. So
in practice, the limit of the computer's main memory becomes the limit on the
size of the BDDs.

\subsection{Related Work}

Prior work has been done to overcome the I/Os spent while computing on BDDs.
Ben-David et al.~\cite{BenDavid2000} and Grumberg, Heyman, and
Schuster~\cite{Grumberg2005} have made distributed symbolic model checking
algorithms that split the set of states between multiple computation nodes. This
makes the BDD on each machine of a manageable size, yet it only moves the
problem from upgrading main memory of a single machine to expanding the number
of machines. David Long~\cite{Long1998} achieved a performance increase of a
factor of two by blocking all nodes in the unique node table based on their time
of creation, i.e.\ with a depth-first blocking. But, in \cite{Arge1996} this was
shown to only improve the worst-case behaviour by a constant. Minato and
Ishihara~\cite{Minato2001} got BDD manipulation to work on disk by serializing
the depth-first traversal of the BDDs, where hash tables in main memory were
used to identify prior visited nodes in the input and output streams. With
limited main memory these tables could not identify all prior constructed nodes
and so the serialization may include the same subgraphs multiple times. This
breaks canonicity of the BDDs and may also in the worst-case result in an
exponential increase of the constructed BDDs.

Ochi, Yasuoka, and Yajima~\cite{Ochi1993} made in $1993$ the BDD manipulation
algorithms breadth-first to thereby exploit a levelwise locality on disk. Their
technique has been heavily improved by Ashar and Cheong~\cite{Ashar1994} in
$1994$ and further improved by Sanghavi et al.~\cite{Sanghavi1996} in $1996$ to
produce the BDD library CAL capable of manipulating BDDs larger than the main
memory. Kunkle, Slavici and Cooperman~\cite{Kunkle2010} extended in $2010$ the
breadth-first approach to distributed BDD manipulation.

The breadth-first algorithms in \cite{Ochi1993,Ashar1994,Sanghavi1996} are not
optimal in the I/O-model, since they still use a single hash table for each
level. This works well in practice, as long as a single level of the BDD can fit
into main memory. If not, they still exhibit the same worst-case I/O behaviour
as other algorithms \cite{Arge1996}.

In 1995, Arge~\cite{Arge1995:1,Arge1996} proposed optimal I/O algorithms for the
basic BDD operations Apply and Reduce. To this end, he dropped all use of hash
tables. Instead, he exploited a total and topological ordering of all nodes
within the graph. This is used to store all recursion requests in priority
queues, s.t.\ they are synchronized with the iteration through the sorted input
stream of nodes. Martin Šmérek attempted to implement these algorithms in 2009
exactly as they were described in \cite{Arge1995:1,Arge1996}. But, the
performance of the resulting external memory symbolic model checking algorithms
was disapointing, since the size of the unreduced BDD and the number of
isomorphic nodes to merge grew too large and unwieldy in practice [personal
communication, Sep 2021].

\subsection{Contributions}

Our work directly follows up on the theoretical contributions of Arge in
\cite{Arge1995:1,Arge1996}. We simplify his I/O-optimal Apply and Reduce
algorithms. In particular, we modify the intermediate representation, to prevent
data duplication and to save on the number of sorting operations. Furthermore,
we are able to prune the output of the Apply and decrease the number of elements
needed to be placed in the priority queue of Reduce, which in practice improves
space use and running time performance. We also provide I/O-efficient versions
of several other standard BDD operations, where we obtain asymptotic
improvements for the operations that are derivable from Apply. Finally, we
reduce the ordering of nodes to a mere ordering of integers and we propose a
priority queue specially designed for these BDD algorithms, to obtain a
considerable constant factor improvement in performance.

Our proposed algorithms and data structures have been implemented to create a
new easy-to-use and open-source BDD package, named \Adiar. Our experimental
evaluation shows that these techniques enable the manipulation of
BDDs larger than the given main memory, with only an acceptable
slowdown compared to a conventional BDD package running exclusively in main
memory.

\subsection{Overview}

The rest of the paper is organised as follows. Section~\ref{sec:preliminaries}
covers preliminaries on the I/O-model and Binary Decision Diagrams. We present
our algorithms for I/O-efficient BDD manipulation in Section~\ref{sec:theory}
together with ways to improve their performance. Section~\ref{sec:adiar}
provides an overview of the resulting BDD package, \Adiar, and
Section~\ref{sec:experiments} contains the experimental evaluation of our
proposed algorithms. Finally, we present our conclusions and future work in
Section~\ref{sec:conclusion}.

\clearpage
\section{Preliminaries} \label{sec:preliminaries}

\subsection{The I/O-Model} \label{sec:preliminaries:io}

The I/O-model \cite{Aggarwal1987} allows one to reason about the number of data
transfers between two levels of the memory hierarchy, while abstracting away from
technical details of the hardware, to make a theoretical analysis manageable.

An I/O-algorithm takes inputs of size $N$, residing on the higher level of the
two, i.e.\ in \emph{external storage} (e.g on a disk). The algorithm can only do
computations on data that reside on the lower level, i.e.\ in \emph{internal
  storage} (e.g.\ main memory). This internal storage can only hold a smaller
and finite number of $M$ elements. Data is transferred between these two levels
in blocks of $B$ consecutive elements \cite{Aggarwal1987}. Here, $B$ is a
constant size not only encapsulating the page size or the size of a cache-line
but more generally how expensive it is to transfer information between the two
levels. The cost of an algorithm is the number of data transfers, i.e.\ the
number of \emph{I/O-operations} or just \emph{I/Os}, it uses.

For all realistic values of $N$, $M$, and $B$ the following inequality holds.
\begin{equation*}
    N/B < \sort(N) \ll N
    \enspace ,
\end{equation*}
where $\sort(N) \triangleq N/B \cdot \log_{M/B} (N/B)$ \cite{Aggarwal1987} is
the sorting lower bound, i.e.\ it takes $\Omega(\sort(N))$ I/Os in the worst-case
to sort a list of $N$ elements \cite{Aggarwal1987}. With an $M/B$-way merge sort
algorithm, one can obtain an optimal $O(\sort(N))$ I/O sorting algorithm
\cite{Aggarwal1987}, and with the addition of buffers to lazily update a tree
structure, one can obtain an I/O-efficient priority queue capable of inserting
and extracting $N$ elements in $O(\sort(N))$ I/Os \cite{Arge1995:2}.

\subsubsection{Cache-Oblivious Algorithms}

An algorithm is cache-oblivious if it is I/O-efficient without explicitly making
any use of the variables $M$ or $B$; i.e.\ it is I/O-efficient regardless of the
specific machine in question and across all levels of the memory hierarchy
\cite{Frigo1999}. We furthermore assume the relationship between $M$ and $B$
satisfies the \emph{tall cache assumption}~\cite{Frigo1999}, which is that
\begin{equation*}
  M = \Omega(B^2)
  \enspace .
\end{equation*}
With a variation of the merge sort algorithm one can make it
cache-oblivious \cite{Frigo1999}. Furthermore, Sanders~\cite{Sanders2001}
demonstrated how to design a cache-oblivious priority queue.

\subsubsection{TPIE}

The \TPIE\ software library~\cite{Vengroff1994} provides an implementation of
I/O-efficient algorithms and data structures such that the management of the
$B$-sized buffers is completely transparent to the programmer.

Elements can be stored in files that act like lists and are accessed with
iterators. Each iterator can \texttt{\bf write} new elements at the end of a
file or on top of prior written elements in the file. For our purposes, we only
need to \texttt{\bf push} elements to a file, i.e.\ \texttt{\bf write} them at
the end of a file. The iterator can traverse the file in both directions by
reading the \texttt{\bf next} element, provided \texttt{\bf has\_next} returns
true. One can also \texttt{\bf peek} the next element without moving the read
head.

\TPIE\ provides an optimal $O(\sort(N))$ external memory merge sort algorithm
for its files. It also provides a merge sort algorithm for non-persistent data
that saves an initial $2 \cdot N/B$ I/Os by sorting the $B$-sized base cases
before flushing them to disk the first time. Furthermore, it provides an
implementation of the I/O-efficient priority queue of \cite{Sanders2001} as
developed in \cite{Petersen2007}, which supports the \texttt{\bf push},
\texttt{\bf top} and \texttt{\bf pop} operations.

\subsection{Binary Decision Diagrams} \label{sec:preliminaries:bdd}

A Binary Decision Diagram (BDD)~\cite{Bryant1986}, as depicted in
Fig.~\ref{fig:bdd_example}, is a rooted directed acyclic graph (DAG) that
concisely represents a boolean function $\B^n \rightarrow \B$, where $\B = \{
\top, \bot \}$. The leaves contain the boolean values $\bot$ and $\top$ that
define the output of the function. Each internal node contains the \emph{label}
$i$ of the input variable $x_i$ it represents, together with two outgoing arcs:
a \emph{low} arc for when $x_i = \bot$ and a \emph{high} arc for when $x_i =
\top$. We only consider Ordered Binary Decision Diagrams (OBDD), where each
unique label may only occur once and the labels must occur in sorted order on
all paths. The set of all nodes with label $j$ is said to belong to the $j$th
\emph{level} in the DAG.

\if\arxiv1
\begin{figure}[b]
  \else
\begin{figure}[t]
  \fi
  \centering

  \subfloat[$x_2$]{
    \label{fig:bdd_example:x2}
    
    \begin{tikzpicture}
      % nodes
      \node[shape = circle, draw = black]                                    (0) {$x_2$};

      % leafs
      \node[shape = rectangle, draw = black, below left=.51cm and .2cm of 0]   (sink_F) {$\bot$};
      \node[shape = rectangle, draw = black, below right=.51cm and .2cm of 0]  (sink_T) {$\top$};

      % arcs
      \draw[->]
      (0) edge (sink_T)
      ;

      \draw[->, dashed]
      (0) edge (sink_F)
      ;      
    \end{tikzpicture}
  }
  \qquad
  \subfloat[$x_0 \wedge x_1$]{
    \label{fig:bdd_example:and}
    
    \begin{tikzpicture}
      % nodes
      \node[shape = circle, draw = black]                                 (0) {$x_0$};
      \node[shape = circle, draw = black, below right=0cm and .7cm of 0] (1) {$x_1$};

      % leafs
      \node[shape = rectangle, draw = black, below=1.37cm of 0]  (sink_F) {$\bot$};
      \node[shape = rectangle, draw = black, below=.9cm of 1]  (sink_T) {$\top$};

      % arcs
      \draw[->]
      (0) edge (1)
      (1) edge (sink_T)
      ;

      \draw[->, dashed]
      (0) edge (sink_F)
      (1) edge (sink_F)
      ;      
    \end{tikzpicture}
  }
  \qquad
  \subfloat[$x_0 \oplus x_1$]{
    \label{fig:bdd_example:xor}
    
    \begin{tikzpicture}
      % nodes
      \node[shape = circle, draw = black]                                 (0) {$x_0$};
      \node[shape = circle, draw = black, below left= 0cm and .5cm of 0] (1) {$x_1$};
      \node[shape = circle, draw = black, below right=0cm and .5cm of 0] (2) {$x_1$};

      % leafs
      \node[shape = rectangle, draw = black, below=.9cm of 1]  (sink_F) {$\bot$};
      \node[shape = rectangle, draw = black, below=.9cm of 2]  (sink_T) {$\top$};

      % arcs
      \draw[->]
      (0) edge (2)
      (1) edge (sink_T)
      (2) edge (sink_F)
      ;

      \draw[->, dashed]
      (0) edge (1)
      (1) edge (sink_F)
      (2) edge (sink_T)
      ;      
    \end{tikzpicture}
  }
  \qquad
  \subfloat[$x_1 \vee x_2$]{
    \label{fig:bdd_example:or}
    
    \begin{tikzpicture}
      % nodes
      \node[shape = circle, draw = black]  (1) {$x_1$};
      \node[shape = circle, draw = black, below left= 0cm and .7cm of 1] (2) {$x_2$};

      % leafs
      \node[shape = rectangle, draw = black, below=.4cm of 2]  (sink_F) {$\bot$};
      \node[shape = rectangle, draw = black, below=.9cm of 1]  (sink_T) {$\top$};

      % arcs
      \draw[->]
      (1) edge (sink_T)
      (2) edge (sink_T)
      ;

      \draw[->, dashed]
      (1) edge (2)
      (2) edge (sink_F)
      ;      
    \end{tikzpicture}
  }
  
  \caption{Examples of Reduced Ordered Binary Decision Diagrams. Leaves are
    drawn as boxes with the boolean value and internal nodes as circles with the
    decision variable. \emph{Low} edges are drawn dashed while \emph{high} edges
    are solid.}
  \label{fig:bdd_example}
\end{figure}

If one exhaustively (1) skips all nodes with identical children and (2) removes
any duplicate nodes then one obtains the \emph{Reduced Ordered Binary Decision
  Diagram} (ROBDD) of the given OBDD. If the variable order is fixed, this
reduced OBDD is a unique canonical form of the function it represents.
\cite{Bryant1986}

The two primary algorithms for BDD manipulation are called Apply and Reduce. The
Apply computes the OBDD $h = f \odot g$ where $f$ and $g$ are OBDDs and $\odot$
is a function $\B \times \B \rightarrow \B$. This is essentially done by
recursively computing the product construction of the two BDDs $f$ and $g$ and
applying $\odot$ when recursing to pairs of leaves. The Reduce applies the two
reduction rules on an OBDD bottom-up to obtain the corresponding ROBDD.
\cite{Bryant1986}

\subsubsection{I/O-Complexity of Binary Decision Diagrams} \label{sec:preliminaries:bdd:io}

Common implementations of BDDs use recursive depth-first procedures that
traverse the BDD and the unique nodes are managed through a hash table
\cite{Karplus1988,Brace1990,Somenzi2015,Lind1999,Dijk2016}. The latter allows
one to directly incorporate the Reduce algorithm of \cite{Bryant1986} within
each node lookup \cite{Minato1990,Brace1990}. They also use a memoisation table
to minimise the number of duplicate computations
\cite{Somenzi2015,Lind1999,Dijk2016}. If the size $N_f$ and $N_g$ of two BDDs
are considerably larger than the memory $M$ available, each recursion request of
the Apply algorithm will in the worst case result in an I/O-operation when
looking up a node within the memoisation table and when following the low and
high arcs \cite{Klarlund1996,Arge1996}. Since there are up to $N_f \cdot N_g$
recursion requests, this results in up to $O(N_f \cdot N_g)$ I/Os in the worst
case. The Reduce operation transparently built into the unique node table with a
\emph{find-or-insert} function can also cause an I/O for each lookup within this
table \cite{Klarlund1996}. This adds yet another $O(N)$ I/Os, where $N$ is the
number of nodes in the unreduced BDD.

For example, the BDD package BuDDy~\cite{Lind1999} uses a linked-list
implementation for its unique node table's buckets. The index to the first node
of the $i$th bucket is stored together with the $i$th node in the table.
Figure~\ref{fig:dfs io} shows the performance of BuDDy solving the
\emph{Tic-Tac-Toe} benchmark for $N=21$ when given variable amounts of memory
(see Section~\ref{sec:experiments} for a description of this benchmark). This
experiment was done on a machine with $8$~GiB of memory and $8$~GiB of swap.
Since a total of $3075$~MiB of nodes are processed and all nodes are placed
consecutively in memory, then all BDD nodes easily fit into the machine's main
memory. That means, the slowdown from $37$ seconds to $49$ minutes in
Fig.~\ref{fig:dfs io} is purely due to random access into swap memory caused by
the \emph{find-or-insert} function of the implicit Reduce.

\begin{figure}[b]
  \centering

  \begin{tikzpicture}
    \begin{axis}[%
      width=0.6\linewidth, height=0.35\linewidth,
      every tick label/.append style={font=\scriptsize},
      % x-axis
      xlabel={Unique node-table size (GiB)},
      xmin={3.8},
      xmax={10.2},
      xtick distance={1},
      xmajorgrids=true,
      % y-axis
      ylabel={Time (seconds)},
      ymin={0},
      ymax={3200},
      ytick distance={1000},
      ymajorgrids=true,
      grid style={dashed,black!12},
      ]

      \draw [white, pattern=north west lines, pattern color=black!30!white]
      (axis cs: 8, 25) rectangle (axis cs: 10.18, 3175);

      \addplot[thick, samples=1, smooth, black, name path=barrier]
      coordinates {(8,25)(8,3175)};

      \node[black] at (axis cs: 7.6, 2800){\tiny{RAM}};
      \node[black] at (axis cs: 8.4, 2800){\tiny{Swap}};

      \addplot+ [style=plot_buddy]
      table {./data/io_buddy_tic_tac_toe_21_time.tex};
    \end{axis}
  \end{tikzpicture}

  \caption{Running time of BuDDy \cite{Lind1999} solving Tic-Tac-Toe for
    $N=21$ on a laptop with $8$ GiB memory and $8$ GiB swap (lower is
    better).}
  \label{fig:dfs io}
\end{figure}

Lars Arge provided a description of an Apply algorithm that is capable of using
only $O(\sort(N_f \cdot N_g))$ I/Os and a Reduce algorithm that uses
$O(\sort(N))$ I/Os~\cite{Arge1995:1,Arge1996}. He also proves these to be
optimal for the level ordering of nodes on disk used by both
algorithms~\cite{Arge1996}. These algorithms do not rely on the value of $M$ or
$B$, so they can be made cache-aware or cache-oblivious when using underlying
sorting algorithms and data structures with those characteristics. We will not
elaborate further on his original proposal, since our algorithms are simpler and
better convey the algorithmic technique used. Instead, we will mention where our
Reduce and Apply algorithms differ from his.

\section{\texorpdfstring{BDD Operations by Time-forward Processing}{BDD Manipulation by Time-forward Processing}} \label{sec:theory}

Our algorithms exploit the total and topological ordering of the internal nodes
in the BDD depicted in (\ref{eq:ordering}) below, where parents precede their
children. It is topological by ordering a node by its \emph{label}, $i : \N$, and
total by secondly ordering on a node's \emph{identifier}, $\mathit{id} : \N$.
This identifier only needs to be unique on each level as nodes are still
uniquely identifiable by the combination of their label and identifier.
\begin{equation}
  \label{eq:ordering}
  (i_1, \mathit{id}_1) < (i_2, \mathit{id}_2)
  \equiv
  i_1 < i_2 \vee (i_1 = i_2 \wedge \mathit{id}_1 < \mathit{id}_2)
\end{equation}
We write the \emph{unique identifier} $(i, \mathit{id}) : \N \times \N$ for a
node as $x_{i, \mathit{id}}$.

BDD nodes do not contain an explicit pointer to their children but instead the
children's unique identifier. Following the same notion, leaf values are stored
directly in the leaf's parents. This makes a node a triple
$\triple{\mathit{uid}}{\mathit{low}}{\mathit{high}}$ where $\mathit{uid} : \N
\times \N$ is its unique identifier and $\mathit{low}$ and $\mathit{high} : (\N
\times \N) + \B$ are its children. The ordering in (\ref{eq:ordering}) is lifted
to compare the \emph{uid}s of two nodes, and so a BDD is represented by a file
with BDD nodes in sorted order. For example, the BDDs in
Fig.~\ref{fig:bdd_example} would be represented as the lists depicted in
Fig.~\ref{fig:bdd_topological}. This ordering of nodes can be exploited with the
\emph{time-forward processing} technique, where recursive calls are not executed
at the time of issuing the request but instead when the element in question is
encountered later in the iteration through the given input file. This is done
with one or more priority queues that follow the same ordering as the input and
deferring recursion by pushing the request into the priority queues.

\if\arxiv1
\begin{figure}[ht!]
  \else
\begin{figure}[b!]
  \fi
  \centering

  \begin{tabular}{r c l c l c l c}
    \ref{fig:bdd_example:x2}:
    & [ & $(x_{2, 0}, \bot, \top)$ & ]
    \\
    \ref{fig:bdd_example:and}:
    & [ & $(x_{0, 0}, \bot, x_{1,0})$ & , & $(x_{1, 0}, \bot, \top)$ & ]
    \\
    \ref{fig:bdd_example:xor}:
    & [ & $(x_{0, 0}, x_{1,0}, x_{1,1})$ & , & $(x_{1, 0}, \bot, \top)$ & , & $(x_{1, 1}, \top, \bot)$ & ]
    \\
    \ref{fig:bdd_example:or}:
    & [ & $(x_{1, 0}, x_{2,0}, \top)$ & , & $(x_{2, 0}, \bot, \top)$ & ]
  \end{tabular}
  
  \caption{In-order representation of BDDs of Fig.~\ref{fig:bdd_example}}
  \label{fig:bdd_topological}
\end{figure}

The Apply algorithm in~\cite{Arge1996} produces an unreduced OBDD, which is
turned into an ROBDD with Reduce. The original algorithms of Arge solely work on
a node-based representation. Arge briefly notes that with an arc-based
representation, the Apply algorithm is able to output its arcs in the order
needed by the following Reduce, and vice versa. Here, an arc is a triple
$\triple{\mathit{source}}{\mathit{is\_high}}{\mathit{target}}$ (written as
$\arc{\mathit{source}}{\mathit{is\_high}}{\mathit{target}}$) where
$\mathit{source} : \N \times \N$, $\mathit{is\_high} : \B$, and $\mathit{target}
: (\N \times \N) + \B$, i.e.\ the \emph{source} contains the unique identifier
of internal nodes while the \emph{target} either contains a unique identifier or
the value of a leaf. We have further pursued this idea of an arc-based
representation and can conclude that the algorithms indeed become simpler and
more efficient with an arc-based output from Apply. On the other hand, we see no
such benefit over the more compact node-based representation in the case of
Reduce. Hence as is depicted in Fig.~\ref{fig:tandem}, our algorithms work in
tandem by cycling between the node-based and arc-based representation.

\if\arxiv1
\begin{figure}[ht!]
  \else
\begin{figure}[tb!]
  \fi
  \centering

  \begin{tikzpicture}[every text node part/.style={align=center}]
    % Boxes
    \draw (0,0) rectangle ++(2,1)
    node[pos=.5]{\texttt{Apply}};
    \draw (5,0) rectangle ++(2,1)
    node[pos=.5]{\texttt{Reduce}};

    % Arcs
    \draw[->] (-0.5,0.8) -- ++(0.5,0)
    node[pos=-1.3]{$f$ \texttt{nodes}};
    \draw[->] (-0.5,0.2) -- ++(0.5,0)
    node[pos=-1.3]{$g$ \texttt{nodes}};

    \draw[blue, ->] (2,0.8) -- ++(3,0)
    node[pos=0.5,above]{\small\color{blue} internal \texttt{arcs}};

    \node at (3.5,0.5) {\textcolor{darkgray}{$f \odot g$ \texttt{arcs}}};
    
    \draw[red, ->] (2,0.2) -- ++(3,0)
    node[pos=0.5,below]{\small\color{red} leaf \texttt{arcs}};

    \draw[->] (7,0.5) -- ++(0.5,0)
    node[pos=2.8]{$f \odot g$ \texttt{nodes}};
  \end{tikzpicture}
  
  \caption{The Apply--Reduce pipeline of our proposed algorithms}
  \label{fig:tandem}
\end{figure}

Notice that our Apply outputs two files containing arcs: arcs to internal nodes
(blue) and arcs to leaves (red). Internal arcs are output at the time of their
target are processed, and since nodes are processed in ascending order, internal
arcs end up being sorted with respect to the unique identifier of their target.
This groups all in-going arcs to the same node together and effectively reverses
internal arcs. Arcs to leaves, on the other hand, are output at the time of
their source is processed, which groups all out-going arcs to leaves together.
These two outputs of Apply represent a semi-transposed graph, which is exactly
of the form needed by the following Reduce. For example, the Apply on the
node-based ROBDDs in Fig.~\ref{fig:bdd_example:x2} and \ref{fig:bdd_example:and}
with logical implication as the operator will yield the arc-based unreduced OBDD
depicted in Fig.~\ref{fig:bdd_transposed}.

\if\arxiv1
\begin{figure}[ht!]
  \else
\begin{figure}[t]
  \fi
  \centering

  \subfloat[Semi-transposed graph. (pairs indicate nodes in
  Fig.~\ref{fig:bdd_example:x2} and \ref{fig:bdd_example:and}, respectively)]{
    \quad % HACK: width
    \begin{tikzpicture}[every text node part/.style={align=center}]
      % nodes
      \node[shape = circle, draw = black]                                   (0)
      {$x_{0,0}$};
      \node[draw=none, above right=-0.1cm and -0.1cm of 0]
      {\small $(x_{2, 0}, x_{0, 0})$};

      \node[shape = circle, draw = black, below right=.1cm and .5cm of 0]  (1)
      {$x_{1,0}$};
      \node[draw=none, above right=-0.1cm and -0.1cm of 1]
      {\small $(x_{2, 0}, x_{1, 0})$};

      \node[shape = circle, draw = black, below left=.1cm and .5cm of 1] (2)
      {$x_{2,0}$};
      \node[draw=none, above left=-0.1cm and -0.1cm of 2]
      {\small $(x_{2, 0}, \bot)$};

      \node[shape = circle, draw = black, below right=.1cm and .5cm of 1] (3)
      {$x_{2,1}$};
      \node[draw=none, above right=-0.1cm and -0.1cm of 3]
      {\small $(x_{2, 0}, \top)$};

      % leafs
      \node[shape = rectangle, draw = black, below=.4cm of 2]  (sink_F) {$\bot$};
      \node[shape = rectangle, draw = black, below=.4cm of 3]  (sink_T) {$\top$};
      
      % arcs
      \draw[blue, ->, dashed]
        (2) edge (0)
        (2) edge (1)
      ;
      \draw[blue, ->]
        (1) edge (0)
        (3) edge (1)
      ;

      \draw[red, ->]
        (2) edge (sink_F)
        (3) edge[bend left] (sink_T)
      ;
      \draw[red, ->, dashed]
        (2) edge (sink_T)
        (3) edge[bend right] (sink_T)
      ;
    \end{tikzpicture}
    \quad % HACK: width
  }
  \quad
  \subfloat[In-order arc-based representation.]{
    \quad % HACK: width
    \begin{tabular}[b]{rcl rcl}
      & \textcolor{blue}{internal \texttt{arcs}} & \hspace{10pt}
      &  & \textcolor{red}{leaf \texttt{arcs}}
      \\ \hline
      [ & \carc{x_{0,0}}{\top}{x_{1,0}}{blue} & , &
                                                    [ & \carc{x_{2,0}}{\bot}{\top}{red} & ,
      \\
      & \carc{x_{1,0}}{\bot}{x_{2,0}}{blue} & , &
         & \carc{x_{2,0}}{\top}{\bot}{red} & ,
      \\
      & \carc{x_{0,0}}{\bot}{x_{2,0}}{blue} & , &
         & \carc{x_{2,1}}{\bot}{\top}{red} & ,
      \\
      & \carc{x_{1,0}}{\top}{x_{2,1}}{blue} & ] &
         & \carc{x_{2,1}}{\top}{\top}{red} & ]
    \end{tabular}
    \quad % HACK: width
  }

  \caption{Unreduced output of Apply when computing $x_2 \Rightarrow (x_0 \wedge x_1)$}
  \label{fig:bdd_transposed}
\end{figure}

For simplicity, we will ignore any cases of leaf-only BDDs in our presentation
of the algorithms. They are easily extended to also deal with those cases.

\subsection{Apply} \label{sec:theory:apply}

Our Apply algorithm works by a single top-down sweep through the input DAGs.
Internal arcs are reversed due to this top-down nature, since an arc between two
internal nodes can first be resolved and output at the time of the arc's target.
These arcs are placed in the file $\Lred$. Arcs from nodes to leaves are placed
in the file $\Lredleaf$.

The algorithm itself essentially works like the standard Apply algorithm. Given
a recursion request for the pair of input nodes $v_f$ from $f$ and $v_g$ from
$g$, a single node $v$ is created with label
$\min(v_f.\mathit{uid}.\mathit{label},v_g.\mathit{uid}.\mathit{label})$ and
recursion requests $r_{\mathit{low}}$ and $r_{\mathit{high}}$ are creatd for its
two children. If the label of $v_f$.$\mathit{uid}$ and $v_g$.$\mathit{uid}$ are
equal, then $r_{\mathit{low}} = (v_f.\mathit{low},v_g.\mathit{low})$ and
$r_{\mathit{high}} = (v_f.\mathit{high},v_g.\mathit{high})$. Otherwise,
$r_{\mathit{low}}$, resp.\ $r_{\mathit{high}}$, contains the $\mathit{uid}$ of
the low child, resp.\ the high child, of $\min(v_f,v_g)$, whereas
$\max(v_f.\mathit{uid},v_g.\mathit{uid})$ is kept as is.

The \texttt{\bf RequestsFor} function computes the requests $r_{\mathit{low}}$
and $r_\mathit{high}$ described above as shown in
Fig.~\ref{fig:apply__requests_for1}. It resolves the request for the tuple
$(t_f, t_g)$ where at least one of them identify an internal node. The arguments
$v_f$ and $v_g$ are the nodes currently read, where we can only guarantee $t_f =
v_f$.$\mathit{uid}$ or $t_g = v_g$.$\mathit{uid}$ but not necessarily that both
match. If the label of $t_f$ and $t_g$ are the same but only one of them
matches, then $\mathit{low}$ and $\mathit{high}$ contains the children of the
other.

\begin{figure}[ht!]
  \centering

\begin{lstlisting}[frame=tb]
RequestsFor(($t_f$, $t_g$), $v_f$, $v_g$, $\mathit{low}$, $\mathit{high}$, $\odot$)
  $r_{\mathit{low}}$, $r_{\mathit{high}}$

  // Combine accumulated information for requests
  if $t_g \in \{ \bot, \top \}$ $\lor$ ($t_f \not\in \{ \bot, \top \}$ $\land$ $t_f$.$\mathit{label}$ $<$ $t_g$.$\mathit{label}$)
  then $r_{\mathit{low}}$ :=    ($v_f$.$\mathit{low}$, $t_g$); $r_{\mathit{high}}$ :=          ($v_f$.$\mathit{high}$, $t_g$)
  else if $t_f \in \{ \bot, \top \}$ $\lor$ ($t_f$.$\mathit{label}$ $>$ $t_g$.$\mathit{label}$)
  then $r_{\mathit{low}}$ :=    ($t_f$, $v_g$.$\mathit{low}$); $r_{\mathit{high}}$ :=          ($t_f$, $v_g$.$\mathit{high}$)
  else
    if $t_f$ = $v_f$.$\mathit{uid}$ $\land$ $t_g$ = $v_g$.$\mathit{uid}$
    then $r_{\mathit{low}}$ :=    ($v_f$.$\mathit{low}$, $v_g$.$\mathit{low}$); $r_{\mathit{high}}$ :=            ($v_f$.$\mathit{high}$, $v_g$.$\mathit{high}$)
    else if $t_f$ = $v_f$.uid
    then $r_{\mathit{low}}$ :=    ($v_f$.$\mathit{low}$, $\mathit{low}$); $r_{\mathit{high}}$ :=           ($v_f$.$\mathit{high}$, $\mathit{high}$)
    else $r_{\mathit{low}}$ :=    ($\mathit{low}$, $v_g$.$\mathit{low}$); $r_{\mathit{high}}$ :=           ($\mathit{high}$, $v_g$.$\mathit{high}$)

  // Apply $\odot$ if need be
  if $r_{\mathit{low}}[0] \in \{ \bot, \top \}$ $\land$ $r_{\mathit{low}}[1] \in \{ \bot, \top \}$
  then $r_{\mathit{low}}$ :=    $r_{\mathit{low}}[0] \odot r_{\mathit{low}}[1]$

  if $r_{\mathit{high}}[0] \in \{ \bot, \top \}$ $\land$ $r_{\mathit{high}}[1] \in \{ \bot, \top \}$
  then $r_{\mathit{high}}$ :=    $r_{\mathit{high}}[0] \odot r_{\mathit{high}}[1]$

  return $r_{\mathit{low}}$, $r_{\mathit{high}}$
\end{lstlisting}
  
  \caption{The \texttt{RequestsFor} subroutine for Apply}
  \label{fig:apply__requests_for1}
\end{figure}

The surrounding pieces of the Apply algorithm shown in Fig.~\ref{fig:apply} are
designed such that all the information needed by \texttt{\bf RequestsFor} is
made available in an I/O-efficient way. Input nodes $v_f$ and $v_g$ are read on
lines 12~--~14 from $f$ and $g$ as a merge of the two sorted lists based on the
ordering in (\ref{eq:ordering}). Specifically, these lines maintain that
$t_{\mathit{seek}} \leq v_f$ and $t_{\mathit{seek}} \leq v_g$, where the
$t_{\mathit{seek}}$ variable itself is a monotonically increasing unique
identifier that changes based on the given recursion requests for a pair of
nodes $(t_f,t_g)$. These requests are synchronised with this traversal of the
nodes by use of the two priority queues $\QappA$ and $\QappB$. $\QappA$ has
elements of the form $(\arc{s}{\mathit{is\_high}}{(t_f,t_g)})$ while $\QappB$
has elements of the form $(\arc{s}{\mathit{is\_high}}{(t_f,t_g)}, \mathit{low},
\mathit{high})$. The use for the boolean $\mathit{is\_high}$ and the unique
identifiers $s$, $\mathit{low}$, and $\mathit{high}$ will become apparent below.

\begin{figure}[p]
  \centering

\begin{lstlisting}[frame=tb]
Apply($f$, $g$, $\odot$)
  $\Lred$ :=      []; $\Lredleaf$ :=    []; $\QappA$ :=      $\emptyset$; $\QappB$ :=       $\emptyset$
  $v_f$ :=   $f$.next(); $v_g$ :=    $g$.next(); id :=  0; label := undefined

  /* Insert request for root $(v_f,v_g)$ */
  $\QappA$.push($\arc{\texttt{NIL}}{\mathit{undefined}}{(v_f.\mathit{uid}, v_g.\mathit{uid})}$)

  /* Process requests in topological order */
  while $\QappA \neq \emptyset \lor \QappB \neq \emptyset$ do
    ($\arc{s}{\mathit{is\_high}}{(t_f, t_g)}$, $\mathit{low}$, $\mathit{high}$) :=                  TopOf($\QappA$, $\QappB$)

    $t_{\mathit{seek}}$ :=    if $\mathit{low}$, $\mathit{high}$ = NIL then min($t_f$,$t_g$) else max($t_f$,$t_g$)
    while $v_f$.$\mathit{uid}$ < $t_{\mathit{seek}}$ $\land$ $f$.has_next() do $v_f$ :=            $f$.next() od
    while $v_g$.$\mathit{uid}$ < $t_{\mathit{seek}}$ $\land$ $g$.has_next() do $v_g$ :=            $g$.next() od

    if $\mathit{low}$ = NIL $\land$ $\mathit{high}$ = NIL $\land$ $t_f \not\in \{ \bot, \top \}$ $\land$ $t_g \not\in \{ \bot, \top \}$
                 $\land$ $t_f$.$\mathit{label}$ = $t_g$.$\mathit{label}$ $\land$ $t_f$.$\mathit{id}$ $\neq$ $t_g$.$\mathit{id}$
    then /* Forward information of $\min(t_f,t_g)$ to $\max(t_f,t_g)$ */
      $v$ :=  if $t_{\mathit{seek}} = t_f$ then $v_f$ else $v_g$
      while $\QappA$.top() matches $\arc{\_}{}{(t_f,t_g)}$ do
        ($\arc{s}{\mathit{is\_high}}{(t_f, t_g)}$) :=            $\QappA$.pop()
        $\QappB$.push($\arc{s}{\mathit{is\_high}}{(t_f,t_g)}$, $v$.$\mathit{low}$, $v$.$\mathit{high}$)
      od
    else /* Process request $(t_f,t_g)$ */
      id := if label $\neq$ $t_{\mathit{seek}}$.$\mathit{label}$ then 0 else id+1
      label := $t_{\mathit{seek}}$.$\mathit{label}$

      /* Forward or output out-going arcs */
      $r_{\mathit{low}}$, $r_{\mathit{high}}$ :=       RequestsFor(($t_f$, $t_g$), $v_f$, $v_g$, $\mathit{low}$, $\mathit{high}$, $\odot$)
      (if $r_{\mathit{low}} \in \{ \bot, \top \}$ then $\Lredleaf$ else $\QappA$).push($\arc{x_{\texttt{label},{\texttt{id}}}}{\bot}{r_{\mathit{low}}}$)
      (if $r_{\mathit{high}} \in \{ \bot, \top \}$ then $\Lredleaf$ else $\QappA$).push($\arc{x_{\texttt{label},{\texttt{id}}}}{\top}{r_{\mathit{high}}}$)

      /* Output in-going arcs */
      while $\QappA \neq \emptyset$ $\land$ $\QappA$.top() matches ($\arc{\_}{}{(t_f,t_g)}$) do
        ($\arc{s}{\mathit{is\_high}}{(t_f, t_g)}$) :=            $\QappA$.pop()
        if $s$ $\neq$ NIL then $\Lred$.push($\arc{s}{\mathit{is\_high}}{x_{\texttt{label},\texttt{id}}}$)
      od
      while $\QappB \neq \emptyset$ $\land$ $\QappB$.top() matches ($\arc{\_}{}{(t_f,t_g)}$, _, _) do
        ($\arc{s}{\mathit{is\_high}}{(t_f, t_g)}$, _, _) :=            $\QappB$.pop()
        if $s$ $\neq$ NIL then $\Lred$.push($\arc{s}{\mathit{is\_high}}{x_{\texttt{label},\texttt{id}}}$)
      od
  od
  return $\Lred$, $\Lredleaf$
\end{lstlisting}

  \caption{The Apply algorithm}
  \label{fig:apply}
\end{figure}

The priority queue $\QappA$ is aligned with the node $\min(t_f,t_g)$, i.e.\ the
one of $t_f$ and $t_g$ that is encountered first. That is, its elements are
sorted in ascending order based on $\min(t_f,t_g)$ of each request. Requests to
the same $(t_f,t_g)$ are grouped together by secondarily sorting the tuples
lexicographically. The algorithm maintains the following invariant between the
current nodes $v_f$ and $v_g$ and the requests within $\QappA$.
\begin{equation*}
  \forall\ (\arc{s}{\mathit{is\_high}}{(t_f,t_g)}) \in \QappA : v_f \leq t_f \wedge v_g \leq t_g
  \enspace .
\end{equation*}
The second priority queue $\QappB$ is used in the case of $t_f.\mathit{label} =
t_g.\mathit{label}$ and $t_f.\mathit{id} \neq t_g.\mathit{id}$, i.e.\ when
\texttt{\bf RequestsFor} needs information from both nodes to resolve the request
but they are not guaranteed to be visited simultaneously. To this end, it is
sorted by $\max(t_f,t_g)$ in ascending order, i.e.\ the second of the two to be
visited, and ties are again broken lexicographically. The invariant for this
priority queue is comparatively more intricate.
\begin{align*}
  \forall (\arc{s}{\mathit{is\_high}}{(t_f,t_g)}, \mathit{low}, \mathit{high}) \in \QappB\ :
    &\ t_f.\mathit{label} = t_g.\mathit{label} = \min(v_f,v_g).\mathit{label}
  \\
    & \wedge t_f.\mathit{id} \neq t_g.\mathit{id}
  \\
    & \wedge t_f < t_g \implies t_f \leq v_g \leq t_g
  \\
    & \wedge t_g < t_f \implies t_g \leq v_f \leq t_f
\end{align*}
In the case that a request is made for a tuple $(t_f,t_g)$ with the same label
but different identifiers, then they are not necessarily available at the same
time. Rather, the node with the unique identifier $\min(t_f,t_g)$ is visited
first and the one with $\max(t_f,t_g)$ some time later. Hence, requests
$\arc{s}{\mathit{is\_high}}{(t_f,t_g)}$ are moved on lines 19~--~23 from
$\QappA$ into $\QappB$ when $\min(t_f,t_g)$ is encountered. Here, the request is
extended with the \emph{low} and \emph{high} of the node $\min(v_f,v_g)$ such
that the children of $\min(v_f,v_g)$ are available at $\max(v_f,v_g)$, despite
the fact that $\min(v_f,v_g).\mathit{uid} < \max(v_f,v_g).\mathit{uid}$.

The requests from $\QappA$ and $\QappB$ are merged on lines 10 with the
\texttt{\bf TopOf} function such that they are synchronised with the order of nodes
in $f$ and $g$ and maintain the above invariants. If both $\QappA$ and $\QappB$
are non-empty, then let $r_1 = (\arc{s_1}{b_1}{(t_{f:1},t_{g:1})})$ be the top
element of $\QappA$ and let the top element of $\QappB$ be $r_2 =
(\arc{s_2}{b_2}{(t_{f:2}, t_{g:2})}, \mathit{low}, \mathit{high})$. Then
$\texttt{\bf TopOf}(\QappA,\QappB)$ outputs $(r_1, \texttt{Nil}, \texttt{Nil})$ if
$\min(t_{f:1},t_{g:1}) < \max(t_{f:2},t_{g:2})$ and $r_2$ otherwise. If either
one is empty, then it equivalently outputs the top request of the other.

When a request is resolved, then the newly created recursion requests
$r_{\mathit{low}}$ and $r_{\mathit{high}}$ to its children are placed at the end
of $\Lredleaf$ or pushed into $\QappA$ on lines 28~--~31, depending on whether
it is a request to a leaf or not. All ingoing arcs to the resolved request are
output on lines 33~--~41. To output these ingoing arcs the algorithm uses the
last two pieces of information that was forwarded: the unique identifier $s$ for
the source of the request and the boolean $\mathit{is\_high}$ for whether the
request was following a high arc.

The arc-based output greatly simplifies the algorithm compared to the original
proposal of Arge in~\cite{Arge1996}. Our algorithm only uses two priority queues
rather than four. Arge's algorithm, like ours, resolves a node before its
children, but due to the node-based output it has to output this entire node
before its children. Hence, it has to identify its children by the tuple
$(t_f,t_g)$ which doubles the amount of space used and it also forces one to
relabel the graph afterwards costing yet another $O(\sort(N))$ I/Os. Instead,
the arc-based output allows us to output the information at the time of the
children and hence we are able to generate the label and its new identifier for
both parent and child. Arge's algorithm neither forwarded the source $s$ of a
request, so repeated requests to the same pair of nodes were merely discarded
upon retrieval from the priority queue, since they carried no relevant
information. Our arc-based output, on the other hand, makes every element placed
in the priority queue forward a source $s$, vital for the creation of the
transposed graph.

\begin{proposition}[Following Arge 1996~\cite{Arge1996}] \label{prop:apply}
  The Apply algorithm in Fig.~\ref{fig:apply} has I/O complexity $O(\sort(N_f
  \cdot N_g))$ and $O((N_f\cdot N_g) \cdot \log (N_f \cdot N_g))$ time
  complexity, where $N_f$ and $N_g$ are the respective sizes of the BDDs for $f$
  and $g$.
\end{proposition}
\begin{proof}
  It only takes $O((N_f + N_g) / B)$ I/Os to read the elements of $f$ and $g$
  once in order. There are at most $2 \cdot N_f \cdot N_g$ many arcs being
  outputted into $\Lred$ or $\Lredleaf$, resulting in at most $O((N_f \cdot N_g)
  / B)$ I/Os spent on writing the output. Each element creates up to two
  requests for recursions placed in $\QappA$, each of which may be reinserted in
  $\QappB$ to forward data across the level, which totals $O(\sort(N_f \cdot
  N_g))$ I/Os. All in all, an $O(\sort(N_f \cdot N_g))$ number of I/Os are used.

  The worst-case time complexity is derived similarly, since next to the
  priority queues and reading the input only a constant amount of work is done
  per request.
\end{proof}

\subsubsection{Pruning by Short Circuiting the Operator} \label{sec:theory:apply:pruning}

The Apply procedure as presented above, like Arge's original algorithm in
\cite{Arge1995:2,Arge1996}, follows recursion requests until a pair of leaves
are met. Yet, for example in Fig.~\ref{fig:bdd_transposed} the node for the
request $(x_{2,0},\top)$ is unnecessary to resolve, since all leaves of this
subgraph trivially will be $\top$ due to the implication operator. The later
Reduce will remove any nodes with identical children bottom-up and so this node
will be removed in favour of the $\top$ leaf.

This observation implies we can resolve the requests earlier for any operator
that can short circuit the resulting leaf. To this end, one only needs to extend the
\texttt{\bf RequestsFor} to be the \texttt{\bf ShortCircuitedRequestsFor} function
in Fig.~\ref{fig:apply__requests_for2} which immediately outputs a request to
the relevant leaf value. This decreases the number of requests placed in
$\QappA$ and $\QappB$. Furthermore, it decreases the size of the final unreduced
BDD, which again in turn will speed up the following Reduce.

\begin{figure}[ht!]
  \centering

\begin{lstlisting}[frame=tb]
ShortCircuitingRequestsFor(($t_f$, $t_g$), $v_f$, $v_g$, $\mathit{low}$, $\mathit{high}$, $\odot$)
  ($r_{\mathit{low}}$, $r_{\mathit{high}}$) :=       RequestsFor(($t_f$, $t_g$), $v_f$, $v_g$, $\mathit{low}$, $\mathit{high}$, $\odot$)

  if $r_{\mathit{low}}$ matches $(t_f', t_g')$
  then if $t_f' \in \{ \bot, \top \}$ $\land$ $t_f' \odot \bot = t_f' \odot \top$
       then $r_{\mathit{low}}$ :=    $t_f' \odot \bot$
       
       if $t_g' \in \{ \bot, \top \}$ $\land$ $\bot \odot t_g' = \top \odot t_g'$
       then $r_{\mathit{low}}$ :=    $\bot \odot t_g'$

  if $r_{\mathit{high}}$ matches $(t_f', t_g')$
  then if $t_f' \in \{ \bot, \top \}$ $\land$ $t_f' \odot \bot = t_f' \odot \top$
       then $r_{\mathit{high}}$ :=    $t_f' \odot \bot$
       
       if $t_g' \in \{ \bot, \top \}$ $\land$ $\bot \odot t_g' = \top \odot t_g'$
       then $r_{\mathit{high}}$ :=    $\bot \odot t_g'$
  
  return $r_{\mathit{low}}$, $r_{\mathit{high}}$
\end{lstlisting}
  
  \caption{The \texttt{ShortCircuitingRequestsFor} subroutine for Apply}
  \label{fig:apply__requests_for2}
\end{figure}

\subsection{Reduce} \label{sec:theory:reduce}

Our Reduce algorithm in Fig.~\ref{fig:reduce} works like other explicit variants
with a single bottom-up sweep through the unreduced OBDD. Since the nodes are
resolved and output in a bottom-up descending order then the output is exactly
in the reverse order as it is needed for any following Apply. We have so far
ignored this detail, but the only change necessary to the Apply algorithm in
Section~\ref{sec:theory:apply} is for it to read the list of nodes of $f$ and
$g$ in reverse.

\if\arxiv1
\begin{figure}[ht!]
  \else
\begin{figure}[t!]
  \fi
  \centering

\begin{lstlisting}[frame=tb]
Reduce($\Lred$, $\Lredleaf$)
  $\Lout$ := [];  $\Qred$ :=     $\emptyset$

  /* Process each level bottom-up (topological order) */
  while $\Qred \neq \emptyset$ $\lor$ $\Lredleaf$.has_next() do
    j  := PeekMax($\Qred$, $\Lredleaf$).source.label
    $\Lj$ :=  []; $\LjredA$ :=   []; $\LjredB$ :=   []

    /* Obtain nodes on level $j$, filtering redundant nodes */
    while PeekMax($\Qred$, $\Lredleaf$).source.label = j do
      $\ehigh$ :=    PopMax($\Qred$, $\Lredleaf$)
      $\elow$ :=    PopMax($\Qred$, $\Lredleaf$)
      if $\ehigh$.target = $\elow$.target
      then $\LjredA$.push([$\elow$.source $\mapsto$ $\elow$.target])
      else $\Lj$.push(($\elow$.source, $\elow$.target, $\ehigh$.target))
    od

    /* Output reduced nodes, remapping duplicates */
    sort $v \in \Lj$ by $v$.low and secondly by $v$.high
    id := MAX_ID;
    $v'$ :=  undefined
    for each $v \in \Lj$ do
      if $v'$ is undefined or $v$.low $\neq$ $v'$.low or $v$.high $\neq$ $v'$.high
      then
        $v'$ :=   ($x_{j,\texttt{id}}$, $v$.low, $v$.high)
        $\Lout$.push($v'$)
        id := id - 1
      $\LjredB$.push([$v$.uid $\mapsto$ $v'$.uid])
    od

    /* Forward output remapping to parents */
    sort [$\mathit{uid} \mapsto \mathit{uid'}$]$\in \LjredB$ by $\mathit{uid}$ in descending order
    for each [$\mathit{uid} \mapsto \mathit{uid'}$] $\in$ MergeMaxUid($\LjredA$, $\LjredB$) do
      while arcs from $\Lred$.peek() matches $\arc{\_}{\_}{\mathit{uid}}$ do
        ($\arc{s}{\mathit{is\_high}}{\mathit{uid}}$) :=          $\Lred$.next()
        $\Qred$.push($\arc{s}{\mathit{is\_high}}{\mathit{uid'}}$)
      od
    od
  od
  return $\Lout$
\end{lstlisting}

  \caption{The Reduce algorithm}
  \label{fig:reduce}
\end{figure}

The priority queue $\Qred$ is used to forward the reduction result of a node $v$
to its parents in an I/O-efficient way. $\Qred$ contains arcs from unresolved
sources $s$ in the given unreduced OBDD to already resolved targets $t'$ in the
ROBDD under construction. The bottom-up traversal corresponds to resolving all
nodes in descending order. Hence, arcs $\arc{s}{\mathit{is\_high}}{t'}$ in
$\Qred$ are first sorted on $s$ and secondly on \emph{is\_high}; the latter
simplifies retrieving low and high arcs on lines 11 and 12, since the high arc
is always retrieved first. The base-cases for the Reduce algorithm are the arcs
to leaves in $\Lredleaf$, which follow the exact same ordering. Hence, on lines
11 and 12, arcs in $\Qred$ and $\Lredleaf$ are merged using the \texttt{\bf
  PopMax} function that retrieves the arc that is maximal with respect to this
ordering.

Since nodes are resolved in descending order, $\Lred$ follows this ordering on
the arc's target when elements are read in reverse. The reversal of arcs in
$\Lred$ makes the parents of a node $v$, to which the reduction result is to be
forwarded, readily available on lines 32~--~38.

The algorithm otherwise proceeds similarly to other Reduce algorithms. For each
level $j$ (obtained via \texttt{\bf PeekMax} to look at the maximal unprocessed
arc, if any), all nodes $v$ of that level are created from their high and low
arcs, $\ehigh$ and $\elow$, taken out of $\Qred$ and $\Lredleaf$. The nodes are
split into the two temporary files $\LjredA$ and $\LjredB$ that contain the
mapping $[\mathit{uid} \mapsto \mathit{uid}']$ from the unique identifier
$\mathit{uid}$ of a node in the given unreduced BDD to the unique identifier
$\mathit{uid'}$ of the equivalent node in the output. $\LjredA$ contains the
mapping from a node $v$ removed due to the first reduction rule and is populated
on lines 10~--~16: if both children of $v$ are the same then the mapping
$[v.\mathit{uid} \mapsto v.\mathit{low}]$ is pushed to $\LjredA$. That is, the
node $v$ is not output and is made equivalent to its single child. $\LjredB$
contains the mappings for the second rule and is resolved on lines 19~--~29.
Here, the remaining nodes are placed in an intermediate file $\Lj$ and sorted by
their children. This makes duplicate nodes immediate successors in $\Lj$. Every
new unique node encountered in $\Lj$ is written to the output $\Lout$ before
mapping itself and all its duplicates to it in $\LjredB$. Since nodes are output
out-of-order compared to the input and without knowing how many will be output
for said level, they are given new decreasing identifiers starting from the
maximal possible value \texttt{MAX\_ID}. Finally, $\LjredB$ is sorted back in
order of $\Lred$ to forward the results in both $\LjredA$ and $\LjredB$ to their
parents on lines 32~--~38. Here, \texttt{\bf MergeMaxUid} merges the mappings
$[\mathit{uid} \mapsto \mathit{uid}']$ in $\LjredA$ and $\LjredB$ by always
taking the mapping with the largest $\mathit{uid}$ from either file.

Since the original algorithm of Arge in \cite{Arge1996} takes a node-based OBDD
as an input and only internally uses node-based auxiliary data structures, his
Reduce algorithm had to create two copies of the input to create the transposed
subgraph: one where the copies were sorted by the nodes' low child and one where
they are sorted by their high child. The reversal of arcs in $\Lred$ merges
these auxiliary data structures of Arge into a single set of arcs easily
generated by the preceding Apply. This not only simplifies the algorithm but
also more than halves the memory used and it eliminates two expensive sorting
steps. We also apply the first reduction rule before the sorting step, thereby
decreasing the number of nodes involved in the remaining expensive computation
of that level.

Another consequence of his node-based representation is that his algorithm had
to move all arcs to leaves into $\Qred$ rather than merging requests from
$\Qred$ with the base-cases from $\Lredleaf$. Let $N_\ell \leq 2N$ be the number
of arcs to leaves then all internal arcs is $2N-N_\ell$, where $N$ is the total
number of nodes. Our semi-transposed input considerably decreases the number
I/Os used and time spent, which makes it worthwhile in practice.

\begin{lemma} \label{prop:separate leafs analysis}
  Use of $\Lredleaf$ saves $\Theta(\sort(N_\ell))$ I/Os.
  % , where $N_\ell \leq 2 \cdot N$ is the number of arcs to leaves and $N$ the number of nodes.
\end{lemma}
\begin{proof}
  Without \Lredleaf, the number of I/Os spent on elements in \Qred\ is
  \begin{equation*}
    \Theta(\sort(2N))
    = \Theta(\sort((2N - N_\ell) + N_\ell))
    = \Theta(\sort(2N - N_\ell)) + \Theta(\sort(N_\ell)) 
  \end{equation*}
  whereas with \Lredleaf\ it is $\Theta(\sort(2N - N_\ell)) + N_\ell / B$. It
  now suffices to focus on the latter half concerning $N_\ell$. The constant $c$
  involved in the $\Theta(\sort(N_\ell))$ of the priority queue must be so large
  that the following inequality holds for $N_\ell$ greater than some given
  $n_0$.
  \begin{equation*}
    0
    < N_\ell / B
    \leq c_1 \cdot \sort(N_\ell)
    \leq c_2 \cdot \sort(N_\ell)
    \enspace .
  \end{equation*}
  The difference $c_1 \cdot \sort(N_\ell / B) - N_\ell / B$ is the least number
  of saved I/Os, whereas $c_2 \cdot \sort(N_\ell / B) - N_\ell / B$ is the
  maximal number of saved I/Os. The $O(\sort(N_\ell / B))$ bound follows easily
  as shown below for $N_\ell > n_0$.
  \begin{equation*}
    c_2 \cdot \sort(N_\ell / B) - N_\ell / B
    < c_2 \cdot \sort(N_\ell / B)
  \end{equation*}
  For the $\Omega(\sort(N_\ell / B))$ lower bound, we need to find a $c_1'$ such
  that $c_1' \cdot \sort(N_\ell) \leq c_1 \cdot \sort(N_\ell / B) - N_\ell / B$
  for $N_\ell$ large enough to satisfy the following inequality.
  \begin{equation*}
    c_1' \leq c_1 - \left( \log_{M/B} (N_\ell / B)  \right)^{-1}
  \end{equation*}
  The above is satisfied with $c_1' \triangleq c_1 / 2$ for $N_\ell > \max(n_0,
  B(\tfrac{M}{B})^{2/c_1})$ since then $\left( \log_{M/B} (N_\ell / B)
  \right)^{-1} < \tfrac{c_1}{2}$.
\end{proof}

Depending on the given BDD, this decrease can result in an improvement in
performance that is notable in practice as our experiments in
Section~\ref{sec:experiments:leaf stream} show. Furthermore, the fraction
$N_\ell / (2N)$ only increases when recursion requests are pruned; our
experiments in Section~\ref{sec:experiments:pruning} show that pruning increases
this fraction by a considerable amount.

\begin{proposition}[Following Arge 1996~\cite{Arge1996}] \label{prop:reduce}
  The Reduce algorithm in Fig.~\ref{fig:reduce} has an $O(\sort(N))$ I/O
  complexity and an $O(N \log N)$ time complexity.
\end{proposition}
\begin{proof}
  Up to $2 N$ arcs are inserted in and later again extracted from \Qred, while
  \Lred\ and \Lredleaf\ is scanned once. This totals an $O(\sort(4N) + N/B) =
  O(\sort(N))$ number of I/Os spent on the priority queue. On each level all
  nodes are sorted twice, which when all levels are combined amounts to another
  $O(\sort(N))$ I/Os. One arrives with similar argumentation at the $O(N \log
  N)$ time complexity.
\end{proof}

Arge proved in \cite{Arge1995:1} that this $O(\sort(N))$ I/O complexity is
optimal for the input, assuming a levelwise ordering of nodes (See
\cite{Arge1996} for the full proof).

% While the $O(N \log N)$ time is theoretically slower than the $O(N)$
% depth-first approach using a unique node table and an $O(1)$ time hash
% function, one should note the log factor depends on the sorting algorithm and
% priority queue used where I/O-efficient instances have a large $M/B$ log
% factor.

\subsection{Other Algorithms} \label{sec:theory:new algorithms}

Arge only considered Apply and Reduce, since most other BDD operations can be
derived from these two operations together with the Restrict operation, which
fixes the value of a single variable. We have extended the technique to create a
time-forward processed Restrict that operates in $O(\sort(N))$ I/Os, which is
described below. Furthermore, by extending the technique to other operations, we
obtain more efficient BDD algorithms than by computing it using Apply and
Restrict alone.

\subsubsection{Node and Variable Count} \label{sec:theory:new algorithms:meta}

The Reduce algorithm in Section~\ref{sec:theory:reduce} can easily be extended
to also provide the number of nodes $N$ generated and the number of levels $L$.
Using this, it only takes $O(1)$ I/Os to obtain the node count $N$ and the
variable count $L$.

\begin{lemma} \label{prop:varcount} \label{prop:nodecount}
  The number of nodes $N$ and number of variables $L$ occuring in the BDD for
  $f$ is obtainable in $O(1)$ I/Os and time.
\end{lemma}

\subsubsection{Negation} \label{sec:theory:new algorithms:negation}

A BDD is negated by inverting the value in its nodes' leaf children. This is an
$O(1)$ I/O-operation if a \emph{negation flag} is used to mark whether the nodes
should be negated on-the-fly as they are read from the stream. 

\begin{proposition} \label{prop:not}
  Negation has I/O, space, and time complexity $O(1)$.
\end{proposition}
\begin{proof}
  Even though $O(N)$ extra work has to be spent on negating every nodes inside
  of all other procedures, this extra work can be accounted for in the big-O of
  the other operations.
\end{proof}

This is a major improvement over the $O(\sort(N))$ I/Os spent by Apply to
compute $f \oplus \top$, where $\oplus$ is exclusive disjunction; especially
since space can be shared between the BDD for $f$ and $\neg f$.

\subsubsection{Equality Checking} \label{sec:theory:new algorithms:equality}

To check for $f \equiv g$ one has to check the DAG of $f$ being isomorphic to
the one for $g$ \cite{Bryant1986}. This makes $f$ and $g$ trivially inequivalent
when at least one of the following three statements are violated.
\begin{align}
    N_f &= N_g
  &
    L_f &= L_g
  &
    \forall i : L_{f,i} = L_{g,i} \wedge N_{f,i} &= N_{g,i}
                          \enspace ,
    \label{eq:trivial cases}
\end{align}
where $N_f$ and $N_g$ is the number of nodes in the BDD for $f$ and $g$, $L_f$
and $L_g$ the number of levels, $L_{f,i}$ and $L_{g,i}$ the respective labels of
the $i$th level, and $N_{f,i}$ and $N_{g,i}$ the number of nodes on the $i$th
level. This can respectively be checked in $O(1)$, $O(1)$, and $O(L/B)$ number
of I/Os with meta-information easily generated by the Reduce algorithm in
Section~\ref{sec:theory:reduce}.

In what follows, we will address equality checking the non-trivial cases.
Assuming the constraints in (\ref{eq:trivial cases}) we will omit $f$ and the
$g$ in the subscript and just write $N$, $L$, $L_i$, and $N_i$.

\paragraph{An $O(\sort(N))$ equality check.}
If $f \equiv g$, then the isomorphism relates the roots of the BDDs for $f$
and $g$. Furthermore, for any node $v_f$ of $f$ and $v_g$ of $g$, if $(v_f,v_g)$
is uniquely related by the isomorphism then so should $(v_f.\mathit{low},
v_g.\mathit{low})$ and $(v_f.\mathit{high}, v_g.\mathit{high})$. Hence, the
Apply algorithm can be adapted to not output any arcs. Instead, it returns
$\bot$ and terminates early when one of the following two conditions are
violated.
\begin{itemize}
\item The children of the given recursion request $(t_f,t_g)$ should both be a
  leaf or an internal node. Furthermore, pairs of internal nodes should agree on
  the label while pairs of leaves should agree on the value.

\item On level $i$, exactly $N_i$ unique recursion requests should be retrieved
  from the priority queues. If more than $N_i$ are given, then prior recursions
  have related one node of $f$, resp.\ $g$, to two different nodes in $g$, resp.\
  $f$. This implies that $G_f$ and $G_g$ cannot be isomorphic.
\end{itemize}
If no recursion requests fail, then the first condition of the two guarantees
that $f \equiv g$ and so $\top$ is returned. The second condition makes the
algorithm terminate earlier on negative cases and lowers the provable complexity
bound.

\begin{proposition} \label{prop:equality checking:v1}
  If the BDDs for $f$ and $g$ satisfy the constraints in (\ref{eq:trivial
    cases}), then the above equality checking procedure has I/O complexity
  $O(\sort(N))$ and time complexity $O(N \cdot \log N)$, where $N$ is the size
  of both BDDs.
\end{proposition}
\begin{proof}
  Due to the second termination case, no more than $N_i$ unique requests are
  resolved for each level $i$ in the BDDs for $f$ and $g$. This totals at most
  $N$ requests are processed by the above procedure regardless of whether $f
  \equiv g$ or not. This bounds the number of elements placed in the priority
  queues to be $2N$. The two complexity bounds then directly follow by similar
  analysis as given in the proof of Proposition~\ref{prop:apply}
\end{proof}

This is an asymptotic improvement on the $O(\sort(N^2))$ equality checking
algorithm by computing $f \leftrightarrow g$ with the Apply procedure and then
checking whether the reduced BDD is the $\top$ leaf. Furthermore, it is also a
major improvement in the practical efficiency, since it does not need to run the
Reduce algorithm, $s$ and $\mathit{is\_high}$ in the
$\arc{s}{\mathit{is\_high}}{(t_f,t_g)}$ recursion requests are irrelevant and
memory is saved by omitting them, it has no $O(N^2)$ intermediate output, and it
only needs a single request for each $(t_f,t_g)$ to be moved into the second
priority queue (since $s$ and $\mathit{is\_high}$ are omitted). In practice, it
performs even better since it terminates on the first violation.

\paragraph{A $2 \cdot N/B$ equality check.}
One can achieve an even faster equality checking procedure, by strengthening the
constraints on the node ordering. On top of the ordering in (\ref{eq:ordering})
on page \pageref{eq:ordering}, we require leaves to come in-order after internal
nodes:
\begin{equation}
  \forall x_{i,\mathit{id}}\ :\ x_{i,\mathit{id}} < \bot < \top
  \label{eq:ordering leaves}
  \enspace ,
\end{equation}
and we add the following constraint onto internal nodes
$(x_{i,\mathit{id}_1},\mathit{low}_1,\mathit{high}_1)$ and
$(x_{i,\mathit{id}_2},\mathit{low}_2,\mathit{high}_2)$ on level $i$:
\begin{equation}
    \mathit{id}_1 < \mathit{id}_2
  \equiv
    \mathit{low}_1 < \mathit{low}_2 
    \vee (\mathit{low}_1 = \mathit{low_2} \wedge \mathit{high}_1 < \mathit{high_2})
    \label{eq:ordering children}
\end{equation}
This makes the ordering in (\ref{eq:ordering}) into a lexicographical
three-dimensional ordering of nodes based on their label and their children. The
key idea here is, since the second reduction rule of Bryant~\cite{Bryant1986}
removes any duplicate nodes then the \emph{low} and \emph{high} children
directly impose a total order on their parents.

We will say that the identifiers are \emph{compact} if all nodes on level $i$
have identifiers within the interval $[0,N_i)$. That is, the $j$th node on the
$i$th level has the identifier $j-1$.

\begin{proposition} \label{prop:numeric match}

  Let $G_f$ and $G_g$ be two ROBDDs with compact identifiers that respect the
  ordering of nodes in (\ref{eq:ordering}) extended with (\ref{eq:ordering
    leaves},\ref{eq:ordering children}). Then, $f \equiv g$ if and only if for
  all levels $i$ the $j$th node on level $i$ in $G_f$ and $G_g$ match
  numerically.
\end{proposition}
\begin{proof}
  Suppose that all nodes in $G_f$ and $G_g$ match. This means they describe the
  same DAG, are hence trivially isomorphic, and so $f \equiv g$.

  Now suppose that $f \equiv g$, which means that $G_f$ and $G_g$ are
  isomorphic. We will prove the stronger statement that for all levels $i$ the
  $j$th node on the $i$th level in $G_f$ and $G_g$ not only match numerically
  but are also the root to the very same subgraphs. We will prove this by strong
  induction on levels $i$, starting from the deepest level $\ell_{\max}$ up to
  the root with label $\ell_{\min}$.

  For $i = \ell_{\max}$, since an ROBDD does not contain duplicate nodes,
  there exists only one or two nodes in $G_f$ and $G_g$ on this level: one for
  $x_{\ell_{\max}}$ and/or one describing $\neg x_{\ell_{\max}}$. Since $\bot <
  \top$, the extended ordering in (\ref{eq:ordering children}) gives us that the
  node describing $x_{\ell_{\max}}$ comes before $\neg x_{\ell_{\max}}$. The
  identifiers must match due to compactness.

  For $\ell_{\min} \leq i < \ell_{\max}$, assume per induction that for all
  levels $i' > i$ the $j'$th node on level $i'$ describe the same subgraphs.
  Let $v_{f,1}, v_{f,2}, \dots, v_{f,N_i}$ and $v_{g,1}, v_{g,2}, \dots,
  v_{g,N_i}$ be the $N_i$ nodes in order on the $i$th level of $G_f$ and $G_g$.
  For every $j$, the isomorphism between $G_f$ and $G_g$ provides a unique $j'$
  s.t. $v_{f,j}$ and $v_{g,j'}$ are the respective equivalent nodes. From the
  induction hypotheses we have that $v_{f,j}.\mathit{low} =
  v_{g,j'}.\mathit{low}$ and $v_{f,j}.\mathit{high} = v_{g,j'}.\mathit{high}$.
  We are left to show that their identifiers match and that $j = j'$, which due
  to compactness are equivalent statements. The non-existence of duplicate nodes
  guarantees that all other nodes on level $i$ in both $G_f$ and $G_g$ differ in
  either their low and/or high child. The total ordering in (\ref{eq:ordering
    children}) is based on these children alone and so from the induction
  hypothesis the number of predecessors, resp. successors, to $v_{f,j}$ must be
  the same as for $v_{g,j'}$. That is, $j$ and $j'$ are the same.
\end{proof}

The Reduce algorithm in Figure~\ref{fig:reduce} already outputs the nodes on
each level in the order that satisfies the constraint (\ref{eq:ordering leaves})
and (\ref{eq:ordering children}). It also outputs the $j$th node of each level
with the identifier $\texttt{MAX\_ID}-N_i+j$, i.e.\ with \emph{compact}
identifiers shifted by $\texttt{MAX\_ID}-N_i$. Hence,
Proposition~\ref{prop:numeric match} applies to the ROBDDs constructed by this
Reduce and they can be compared in $2 \cdot N/B $ I/Os with a single iteration
through all nodes. This iteration fails at the same pair of nodes, or even
earlier, than the prior $O(\sort(N))$ algorithm. Furthermore, this linear scan
of both BDDs is optimal -- even with respect to the constant -- since any
deterministic comparison procedure has to look at every element of both inputs
at least once.

As is apparent in the proof of Proposition~\ref{prop:numeric match}, the
ordering of internal nodes relies on the ordering of leaf values. Hence, the
negation algorithm in Section~\ref{sec:theory:new algorithms:negation} breaks
this property when the negation flags on $G_f$ and $G_g$ do not match.

\begin{corollary} \label{prop:equality checking:v2}
  If $G_f$ and $G_g$ are outputs of Fig.~\ref{fig:reduce} and their negation
  flag are equal, then equality checking of $f \equiv g$ can be done in $2 \cdot
  N / B$ I/Os and $O(N)$ time, where $N$ is the size of $G_f$ and $G_g$.
\end{corollary}

One could use the Reduce algorithm to make the negated BDD satisfy the
preconditions of Proposition~\ref{prop:numeric match}. However, the
$O(\sort(N))$ equality check will be faster in practice, due to its efficiency
in both time, I/O, and space.

\subsubsection{Path and Satisfiability Count}
The number of paths that lead from the root of a BDD to the $\top$ leaf can be
computed with a single priority queue that forwards the number of paths from the
root to a node $t$ through one of its ingoing arcs. At $t$ these ingoing paths
are accumulated, and if $t$ has a $\top$ leaf as a child then the number of
paths to $t$ added to a total sum of paths. The number of ingoing paths to $t$
is then forwarded to its non-leaf children.

This easily extends to counting the number of satisfiable assignments by knowing
the number of variables of the domain and extending the request in the priority
queue with the number of levels visited.

\begin{proposition} \label{prop:pathcount} \label{prop:satcount}

  Counting the number of paths and satisfying assignments in the BDD for $f$ has
  I/O complexity $O(\sort(N))$ and time complexity $O(N \log N)$, where $N$ is
  the size of the BDD for $f$.
\end{proposition}
\begin{proof}
  Every node of $f$ is processed in-order and all $2N$ arcs are inserted and
  extracted once from the priority queue. This totals $O(\sort(N)) + N/B$
  I/Os.

  For all $N$ nodes there spent is $O(1)$ time on every in-going arc. There are
  a total of $2N$ arcs, so $O(N)$ amount of work is done next to reading the
  input and using the priority queue. Hence, the $O(N \log N)$ time spent on the
  priority queue dominates the asymptotic running time.
\end{proof}

Both of these operations not possible to compute with the Apply operation.

\subsubsection{Evaluating and Obtaining Assignments}

Evaluating the BDD for $f$ given some $\vec{x}$ boils down to following a path
from the root down to a leaf based on $\vec{x}$. The levelized ordering makes it
possible to follow a path starting from the root in a single scan of all nodes,
where irrelevant nodes are ignored.

\begin{lemma} \label{prop:eval}

  Evaluating $f$ for some assignment $\vec{x}$ has I/O complexity $N / B +
  \vec{x} / B$ and time complexity $O(N + \lvert \vec{x} \rvert)$, where $N$ is
  the size of $f$'s BDD.
\end{lemma}

Similarly, finding the lexicographical smallest or largest $\vec{x}$ that makes
$f(\vec{x}) = \top$ also corresponds to providing the trace of one specific path
in the BDD \cite{Knuth2011}.

\begin{lemma} \label{prop:satmin} \label{prop:satmax}

  Obtaining the lexicographically smallest (or largest) assignment $\vec{x}$
  such that $f(\vec{x}) = \top$ has I/O complexity $N / B + \vec{x} / B$ and
  time complexity $O(N)$, where $N$ is the size of the BDD for $f$.
\end{lemma}

If the number of Levels $L$ is smaller than $N/B$ then this uses more I/Os
compared to the conventional algorithms that uses random access; when jumping
from node to node at most one node on each of the $L$ levels is visited, and so
at most one block for each level is retrieved.

\subsubsection{If-Then-Else} \label{sec:theory:new algorithms:ite}
The If-Then-Else procedure is an extension of the idea of the Apply procedure
where requests are triples $(t_f, t_g, t_h)$ rather than tuples and three
priority queues are used rather than two. The idea of pruning in
Section~\ref{sec:theory:apply:pruning} can also be applied here: if $t_f \in \{
\bot, \top \}$ then either $t_g$ or $t_h$ is irrelevant and it can be replaced
with \texttt{Nil} such that these requests are merged into one node, and if $t_g
= t_h \in \{ \bot, \top \}$ then $t_f$ does not need to be evaluated and the
leaf can immediately be output.

\begin{proposition} \label{prop:ite}

  The If-Then-Else procedure has I/O complexity $O(\sort(N_f \cdot N_g \cdot
  N_h))$ and time complexity $O((N_f \cdot N_g \cdot N_h) \log (N_f \cdot N_g
  \cdot N_h))$, where $N_f$, $N_g$, and $N_h$ are the respective sizes of the
  BDDs for $f$, $g$, and $h$.
\end{proposition}
\begin{proof}
  The proof follows by similar argumentation as for Proposition~\ref{prop:apply}.
\end{proof}

This is an $O(N_f)$ factor improvement over the $O(\sort(N_f^2 \cdot N_g \cdot
N_h))$ I/O complexity of using Apply to compute $(f \wedge g) \vee (\neg f
\wedge h)$.

\subsubsection{Restrict}
Given a BDD $f$, a label $i$, and a boolean value $b$ the function $f|_{ x_i =
  b}$ can be computed in a single top-down sweep. The priority queue contains
arcs from a source $s$ to a target $t$. When encountering the target node $t$ in
$f$, if $t$ does not have label $i$ then the node is kept and the arc
$\arc{s}{\mathit{is\_high}}{t}$ is output, as is. If $t$ has label $i$ then the
arc $\arc{s}{\mathit{is\_high}}{t.\mathit{low}}$ is forwarded in the queue if $b
= \bot$ and $\arc{s}{is\_high}{t.\mathit{high}}$ is forwarded if $b = \top$.

This algorithm has one problem with the pipeline in Fig.~\ref{fig:tandem}. Since
nodes are skipped, and the source $s$ is forwarded, then the leaf arcs may end
up being produced and output out of order. In those cases, the output arcs in
$\Lredleaf$ need to be sorted before the following Reduce can begin.

\begin{proposition} \label{prop:restrict}
  Computing the Restrict of $f$ for a single variable $x_i$ has I/O complexity
  $O(\sort(N))$ and time complexity $O(N \log N)$, where $N$, is the size of the
  BDD for $f$.
\end{proposition}
\begin{proof}
  The priority queue requires $O(\sort(2N))$ I/Os since every arc in the
  priority queue is related one-to-one with one of the $2N$ arcs of the input.
  The output consists of at most $2N$ arcs, which takes $2 \cdot N / B$ I/Os to
  write and at most $O(\sort(2N))$ to sort. All in all, the algorithm uses $N/B
  + O(\sort(2N)) + 2 \cdot N / B + O(\sort(2N))$ I/Os.
\end{proof}

This trivially extends to multiple variables in a single sweep by being given
the assignments $\vec{x}$ in ascending order relative to the labels $i$.

\begin{corollary} \label{prop:restrict:multi}
  Computing the Restrict of $f$ for multiple variables $\vec{x}$ has I/O
  complexity $O(\sort(N) + \lvert \vec{x} \rvert / B)$ and time complexity $O(N
  \log N + \lvert \vec{x} \rvert)$, where $N$, is the size of the BDD for $f$.
\end{corollary}
\begin{proof}
  The only extra work added is reading $\vec{x}$ in-order and to decide for all
  $L \leq N$ levels whether to keep the level's nodes or bridge over them. This
  constitutes another $\lvert \vec{x} \rvert / B$ I/Os and $O(N)$ time.
\end{proof}

\subsubsection{Quantification} \label{sec:theory:new algorithms:quantify}
Given $Q \in \{ \forall, \exists \}$, a BDD $f$, and a label $i$ the BDD
representing $Q b \in \B : f|_{x_i = b}$ is computable using $O(\sort(N^2))$
I/Os by combining the ideas for the Apply and the Restrict algorithms.

Requests are either a single node $t$ in $f$ or past level $i$ a tuple of nodes
$(t_1,t_2)$ in $f$. A priority queue initially contains requests from the root
of $f$ to its children. Two priority queues are used to synchronise the tuple
requests with the traversal of nodes in $f$. If the label of a request $t$ is
$i$, then, similar to Restrict, its source $s$ is linked with a single request
to $(t.\mathit{low}, t.\mathit{high})$. Tuples are handled as in Apply where
$\odot$ is $\vee$, resp.\ $\odot$ is $\wedge$, for $Q = \exists$, resp.\ for $Q
= \forall$.

Both conjunction and disjunction are commutative operations, so the order of
$t_1$ and $t_2$ is not relevant. Hence, the tuple $(t_1,t_2)$ can be kept sorted
such that $t_1 < t_2$. This improves the performance of the comparator used
within the priority queues since $\min(t_1,t_2) = t_1$ and $\max(t_1,t_2) =
t_2$.

The idea of pruning in Section~\ref{sec:theory:apply:pruning} also applies here.
But, in the following cases can a request to a tuple be turned back into a
request to a single node $t$.
\begin{itemize}
\item Any request $(t_1,t_2)$ where $t_1 = t_2$ is equivalent to the request to
  $t_1$.
\item As we only consider operators $\vee$ and $\wedge$ then any leaf $v \in
  \{ \bot, \top \}$ that does not short circuit $\odot$ is irrelevant for the leaf
  values of the entire subtree. Hence, any such request $(t,v)$ can safely be
  turned into a request to $t$.
\end{itemize}
This collapses three potential subtrees into one, which decreases the size of
the output and the number of elements placed in the priority queues.

\begin{proposition} \label{prop:forall} \label{prop:exists}

  Computing the Exists and Forall of $f$ for a single variable $x_i$ has I/O
  complexity $O(\sort(N^2))$ and time complexity $O(N^2 \log N^2)$, where $N$
  is the size of the BDD for $f$.
\end{proposition}
\begin{proof}
  A request for a node $t$ is equivalent to a request to $(t,t)$. There are at
  most $N^2$ pairs to process which results in at most $2 \cdot N^2$ number of
  arcs in the output and in the priority queues. The two complexity bounds
  follows by similar analysis as in the proof of Proposition~\ref{prop:apply}
  and \ref{prop:restrict}.
\end{proof}

This could also be computed with Apply and Restrict as $f|_{x_i = \bot} \odot
f|_{x_i = \top}$, which would also only take an $O(\sort(N) + \sort(N) +
\sort(N^2))$ number of I/Os, but this requires three separate runs of Reduce and
makes the intermediate inputs to Apply of up to $2N$ in size.

\subsubsection{Composition}
The composition $f|_{x_i = g}$ of two BDDs $f$ and $g$ can be computed by
reusing the ideas and optimisations from the Quantification and the If-Then-Else
procedures. Requests start out as tuples $(t_g, t_f)$ and at level $i$ they are
turned into triples $(t_g,t_{f_1},t_{f_2})$ which are handled similar to the
If-Then-Else. The idea of merging recursion requests based on the value of
$t_{f_1}$ and $t_{f_2}$ as for Quantification still partially applies here: if
$t_{f_1} = t_{f_2}$ or $t_g$ is a leaf, then the entire triple can boiled down
to $t_{f_1}$ or $t_{f_2}$.

\begin{proposition} \label{prop:composition}
  Composition has I/O complexity $O(\sort(N_f^2 \cdot N_g))$ and time complexity
  $O(N_f^2 \cdot N_g) \log (N_f \cdot N_g))$, where $N_f$ and $N_g$ are the
  respective sizes of the BDDs for $f$ and $g$.
\end{proposition}
\begin{proof}
  There are $N_f^2 \cdot N_g$ possible triples, each of which can be in the
  output and in the priority queue. The two complexity bounds follows by similar
  analysis as in the proof of Proposition~\ref{prop:ite}.
\end{proof}

Alternatively, this is computable with the If-Then-Else $g\ ?\ f|_{x_i = \top}\
:\ f|_{x_i = \bot}$ which would also take $O(2\cdot\sort(N_f) + \sort(N_f^2
\cdot N_g))$ number of I/Os, or $O(\sort(N_f^2 \cdot N_g^2))$ I/Os if the Apply
is used. But as argued above, computing the If-Then-Else with Apply is not
optimal. Furthermore, similar to quantification, the constants involved in using
Restrict for intermediate outputs is also costly.

\subsection{Levelized Priority Queue} \label{sec:theory:priority queue}

In practice, sorting a set of elements with a sorting algorithm is considerably
faster than with a priority queue \cite{Molhave2012}. Hence, the more one can
use mere sorting instead of a priority queue the faster the algorithms will run.

Since all our algorithms resolve the requests level by level then any request
pushed to the priority queues $\QappA$ and $\Qred$ are for nodes on a
yet-not-visited level. That means, any requests pushed when resolving level $i$
do not change the order of any elements still in the priority queue with label
$i$. Hence, for $L$ being the number of levels and $L < M/B$ one can have a
\emph{levelized priority queue}\footnote{The name is chosen as a reference to
  the independently conceived but reminiscent queue of Ashar and
  Cheong~\cite{Ashar1994}. It is important to notice that they split the queue
  for functional correctness whereas we do so to improve the performance.} by
maintaining $L$ buckets and keep one block for each bucket in memory. If a block
is full then it is sorted, flushed from memory and replaced with a new empty
block. All blocks for level $i$ are merge-sorted into the final order of
requests when the algorithm reaches the $i$th level.

While $L < M/B$ is a reasonable assumption between main memory and disk it is
not at the cache-level. Yet, most arcs only span very few levels, so it suffices
for the levelized priority queue to only have a small preset $k$ number of
buckets and then use an \emph{overflow} priority queue to sort the few requests
that go more than $k$ levels ahead. A single block of each bucket fits into the
cache for very small choices of $k$, so a cache-oblivious levelized priority
queue would not lose its characteristics. Simultaneously, smaller $k$ allows one
to dedicate more internal memory to each individual bucket, which allows a
cache-aware levelized priority queue to further improve its performance.

\subsection{Memory Layout and Efficient Sorting} \label{sec:theory:reducing to numbers}

A unique identifier can be represented as a single $64$-bit integer as shown in
Fig.~\ref{fig:bdd_uid_representation}. Leaves are represented with a $1$-flag on
the most significant bit, its value $v$ on the next $62$ bits and a boolean flag
$f$ on the least significant bit. A node is identified with a $0$-flag on the
most significant bit, the next $\ell$-bits dedicated to its label followed by
$62-\ell$ bits with its identifier, and finally the boolean flag $f$ available
on the least-significant bit.

\begin{figure}[ht!]
  \centering

  \subfloat[Unique identifier of a leaf \texttt{v}]{
    \label{fig:bdd_uid_representation:leaf}

    \begin{tikzpicture}
      \fill[black!5!white] (0,0) rectangle ++(5.4,.6);
      
      \draw[black!60!white, pattern=north east lines, pattern color=black!30!white]
      (0,0) rectangle ++(0.2,.6)
      node[pos=.5]{\small\color{black} $1$};

      \draw[black!60!white, pattern=north west lines, pattern color=black!30!white]
      (0.2,0) rectangle ++(5,.6)
      node[pos=.5]{\small\color{black} \texttt{v}};
      
      \draw[black!60!white, pattern=north east lines, pattern color=black!30!white]
      (5.2,0) rectangle ++(0.2,.6)
      node[pos=.5]{\small\color{black} \texttt{f}};
    \end{tikzpicture}
  }
  \quad
  \subfloat[Unique identifier of an internal node]{
    \label{fig:bdd_uid_representation:node}

    \begin{tikzpicture}
      \fill[black!5!white] (0,0) rectangle ++(5.4,.6);
      
      \draw[black!60!white, pattern=north east lines, pattern color=black!30!white]
      (0,0) rectangle ++(0.2,.6)
      node[pos=.5]{\small\color{black} $0$};

      \draw[black!60!white, pattern=north west lines, pattern color=black!30!white]
      (0.2,0) rectangle ++(1.5,.6)
      node[pos=.5]{\small\color{black} \texttt{label}};

      \draw[black!60!white, pattern=north east lines, pattern color=black!30!white]
      (1.7,0) rectangle ++(3.5,.6)
      node[pos=.5]{\small\color{black} \texttt{identifier}};

      \draw[black!60!white, pattern=north west lines, pattern color=black!30!white]
      (5.2,0) rectangle ++(0.2,.6)
      node[pos=.5]{\small\color{black} \texttt{f}};
      
      \draw[black!80!white]
      (0,0) rectangle ++(5.4,.6);
    \end{tikzpicture}
  }

  \caption{Bit-representations. The least significant bit is right-most.}
  \label{fig:bdd_uid_representation}
\end{figure}

A node is represented with $3$ $64$-bit integers: two of them for its children
and the third for its label and identifier. An arc is represented by $2$
$64$-bit numbers: the source and target each occupying $64$-bits and the
\emph{is\_high} flag stored in the $f$ flag of the source. This reduces all the
prior orderings to a mere trivial ordering of $64$-bit integers. A descending
order is used for a bottom-up traversal while the sorting is in ascending order
for the top-down algorithms.

A reasonable value for the number $\ell$ bits dedicated to the label is $24$.
Here there are still $2^{62-\ell} = 2^{38}$ number of nodes available per level,
which is equivalent to $6$~TiB of data in itself.

\section{Adiar: An Implementation} \label{sec:adiar}

\begin{table}[p]
  \centering
  \footnotesize
  
  \begin{tabular}{l | c | c | c}
    \multicolumn{1}{c|}{\Adiar\ function}
                            & Operation                   & I/O complexity                   & Justification
    \\ \hline \hline
    \multicolumn{3}{c}{BDD Base constructors}
    \\ \hline
    \texttt{bdd\_sink(b)}   & $b$                         & $O(1)$
    \\
    \quad\texttt{bdd\_true()}
                            & $\top$                      & 
    \\
    \quad\texttt{bdd\_false()}
                            & $\bot$                      & 
    \\
    \texttt{bdd\_ithvar($i$)}
                            & $x_i$                       & $O(1)$
    \\
    \texttt{bdd\_nithvar($i$)}
                            & $\neg x_i$                  & $O(1)$
    \\
    \texttt{bdd\_and($\vec{i}$)}
                            & $\bigvee_{i \in \vec{i}} x_i$ & $O(k/B)$
    \\
    \texttt{bdd\_or($\vec{i}$)}
                            & $\bigwedge_{i \in \vec{i}} x_i$ & $O(k/B)$
    \\
    \texttt{bdd\_counter($i$,$j$,$t$)}
                            & $\#_{k=i}^j x_k = t$        & $O((i-j) \cdot t / B)$
    \\ \hline
    \multicolumn{3}{c}{BDD Manipulation}
    \\ \hline
    \texttt{bdd\_apply($f$,$g$,$\odot$)}
                            & $f \odot g$                 & $O(\sort(N_f \cdot N_g))$        & Prop.~\ref{prop:apply}, \ref{prop:reduce}
    \\
    \quad\texttt{bdd\_and($f$,$g$)}
                            & $f \wedge g$                & 
    \\
    \quad\texttt{bdd\_nand($f$,$g$)}
                            & $\neg(f \wedge g)$          & 
    \\
    \quad\texttt{bdd\_or($f$,$g$)}
                            & $f \vee g$                  & 
    \\
    \quad\texttt{bdd\_nor($f$,$g$)}
                            & $\neg(f \vee g)$            & 
    \\
    \quad\texttt{bdd\_xor($f$,$g$)}
                            & $f \oplus g$                & 
    \\
    \quad\texttt{bdd\_xnor($f$,$g$)}
                            & $\neg(f \oplus g)$          & 
    \\
    \quad\texttt{bdd\_imp($f$,$g$)}
                            & $f \rightarrow g$           & 
    \\
    \quad\texttt{bdd\_invimp($f$,$g$)}
                            & $f \leftarrow g$            & 
    \\
    \quad\texttt{bdd\_equiv($f$,$g$)}
                            & $f \equiv g$                & 
    \\
    \quad\texttt{bdd\_diff($f$,$g$)}
                            & $f \wedge \neg g$           & 
    \\
    \quad\texttt{bdd\_less($f$,$g$)}
                            & $\neg f \wedge g$           & 
    \\
    \texttt{bdd\_ite($f$,$g$,$h$)}
                            & $f\ ?\ g\ :\ h$             & $O(\sort(N_f \cdot N_g \cdot N_h))$ & Prop.~\ref{prop:ite}, \ref{prop:reduce}
    \\
    \texttt{bdd\_restrict($f$,$i$,$v$)}
                            & $f|_{x_i = v}$ & $O(\sort(N_f))$                                   & Prop.~\ref{prop:restrict}, \ref{prop:reduce}
    \\
    \texttt{bdd\_restrict($f$,$\vec{x}$)}
                            & $f|_{(i,v) \in \vec{x}\ .\ x_i = v}$ & $O(\sort(N_f) + \lvert  \vec{x} \rvert / B)$
                                                                                                & Prop.~\ref{prop:restrict:multi}, \ref{prop:reduce}
    \\
    \texttt{bdd\_exists($f$,$i$)}
                            & $\exists v : f|_{x_i = v}$   & $O(\sort(N_f^2))$                   & Prop.~\ref{prop:exists}, \ref{prop:reduce}
    \\
    \texttt{bdd\_forall($f$,$i$)}
                            & $\forall v : f|_{x_i = v}$   & $O(\sort(N_f^2))$                   & Prop.~\ref{prop:forall}, \ref{prop:reduce}
    \\
    \texttt{bdd\_not($f$)}
                            & $\neg f$                    & $O(1)$                             & Prop.~\ref{prop:not}
    \\ \hline
    \multicolumn{3}{c}{Counting}
    \\ \hline
    \texttt{bdd\_pathcount($f$)}
                            & \#paths in $f$ to $\top$   & $O(\sort(N_f))$                     & Prop.~\ref{prop:pathcount}
    \\
    \texttt{bdd\_satcount($f$)}
                            & $\# \vec{x} : f(\vec{x})$  & $O(\sort(N_f))$                     & Prop.~\ref{prop:satcount}
    \\
    \texttt{bdd\_nodecount($f$)}
                            & $N_f$                      & $O(1)$                              & Lem.~\ref{prop:nodecount}
    \\
    \texttt{bdd\_varcount($f$)}
                            & $L_f$                      & $O(1)$                              & Lem.~\ref{prop:varcount}
    \\ \hline
    \multicolumn{3}{c}{Equivalence Checking}
    \\ \hline
    \texttt{$f$ == $g$}
                            & $f \equiv g$              & $O(\sort(\min(N_f,N_g)))$           & Prop.~\ref{prop:equality checking:v1}
    \\ \hline
    \multicolumn{3}{c}{BDD traversal}
    \\ \hline
    \texttt{bdd\_eval($f$,$\vec{x}$)}
                            & $f(\vec{x})$               & $O((N_f + \lvert \vec{x} \rvert) / B)$ & Lem.~\ref{prop:eval}
    \\
    \texttt{bdd\_satmin($f$)}
                            & $\min \{ \vec{x} \ |\ f(\vec{x}) \}$
                                                          & $O((N_f + \lvert\vec{x}\rvert) / B)$ & Lem.~\ref{prop:satmin}
    \\
    \texttt{bdd\_satmax($f$)}
                            & $\max \{ \vec{x} \ |\ f(\vec{x}) \}$
                                                          & $O((N_f + \lvert\vec{x}\rvert) / B)$ & Lem.~\ref{prop:satmax}
  \end{tabular}

    \caption{Supported operations in \Adiar\ together with their I/O-complexity.
    $N$ is the number of nodes and $L$ is the number of levels in a BDD. An
    assignment $\vec{x}$ is a list of tuples $(i,v) : \N \times \B$.}
  \label{tab:io_features}
\end{table}

The algorithms and data structures described in Section~\ref{sec:theory} have
been implemented in a new BDD package, named \Adiar\footnote{\textbf{adiar}
  $\langle$ portuguese $\rangle$ (\emph{verb})\ :\ to defer, to postpone}, that
supports the operations in Table~\ref{tab:io_features}. The source code for
\Adiar\ is publicly available at
\begin{center}
  \href{https://github.com/ssoelvsten/adiar}{github.com/ssoelvsten/adiar}
\end{center}
and the documentation is available at
\begin{center}
  \href{https://ssoelvsten.github.io/adiar/}{ssoelvsten.github.io/adiar/}
  \enspace .
\end{center}
Interaction with the BDD package is done through C++ programs that include the
\texttt{<adiar/adiar.h>} header file and are built and linked with CMake. Its
two dependencies are the Boost library and the \TPIE\ library; the latter is
included as a submodule of the \Adiar\ repository, leaving it to CMake to build
\TPIE\ and link it to \Adiar.

The BDD package is initialised by calling the \texttt{adiar\_init(memory,
  temp\_dir)} function, where \texttt{memory} is the memory (in bytes) dedicated
to \Adiar\ and \texttt{temp\_dir} is the directory where temporary files will be
placed, which could be a dedicated harddisk. After use, the BDD package is
deinitialised and its given memory is freed by calling the
\texttt{adiar\_deinit()} function.

The \texttt{bdd} object in \Adiar\ is a container for the underlying files to
represent each BDD. A \texttt{\_\_bdd} object is used for possibly unreduced
arc-based OBDD outputs. Both objects use reference counting on the underlying
files to both reuse the same files and to immediately delete them when the
reference count decrements $0$. By use of implicit conversions between the
\texttt{bdd} and \texttt{\_\_bdd} objects and an overloaded assignment operator,
this garbage collection happens as early as possible, making the number of
concurrent files on disk minimal.

A BDD can also be constructed explicitly with the \texttt{node\_writer} object
in $O(N/B)$ I/Os by supplying it all nodes in reverse of the ordering in
(\ref{eq:ordering}).

\section{Experimental Evaluation} \label{sec:experiments}

We assert the viability of our techniques by investigating the following questions.

\begin{enumerate}
\item\label{rq:large instances} How well does our technique perform on BDDs
  larger than main memory?
  
\item\label{rq:improvements} How big an impact do the optimisations
  we propose have on the computation time and memory usage?

  \begin{enumerate}
  \item\label{rq:priority queue} Use of the levelized priority queue.

  \item\label{rq:leaf stream} Separating the node-to-leaf arcs of Reduce from
    the priority queue.
    
  \item\label{rq:pruning} Pruning the output of Apply by short circuiting the given
    operator.

  \item\label{rq:equality} Use of our improved equality checking algorithms.
  \end{enumerate}
    
\item\label{rq:performance comparison} How well do our algorithms perform in
  comparison to conventional BDD libraries that use depth-first recursion, a
  unique table, and caching?
\end{enumerate}

\subsection{Benchmarks}

To evaluate the performance of our proposed algorithms we have created
implemented multiple benchmarks for \Adiar\ and other BDD packages, where the
BDDs are constructed in a similar manner and the same variable ordering is used.
The source code for all benchmarks is publicly available at the following url:
\begin{center}
  \href{https://github.com/ssoelvsten/bdd-benchmark}{github.com/ssoelvsten/bdd-benchmark}
\end{center}
The raw data and data analysis has been made available at
\cite{Soelvsten2021:Artifact:arXiv}.

\paragraph{N-Queens.} The solution to the $N$-Queens problem is the number
of arrangements of $N$ queens on an $N \times N$ board, such that no queen is
threatened by another. Our benchmark follows the description in
\cite{Kunkle2010}: the variable $x_{ij}$ represents whether a queen is placed on
the $i$th row and the $j$th column and so the solution corresponds to the number of
satisfying assignments of the formula

\begin{equation*}
  \bigwedge_{i=0}^{N-1} \bigvee_{j=0}^{N-1} \left( x_{ij} \wedge \neg \mathit{has\_threat}(i,j) \right)
  \enspace ,
\end{equation*}
where $\mathit{has\_threat}(i,j)$ is true, if a queen is placed on a tile
$(k,l)$, that would be in conflict with a queen placed on $(i,j)$. The ROBDD of the
innermost conjunction can be directly constructed, without any BDD operations.

\paragraph{Tic-Tac-Toe. } In this problem from \cite{Kunkle2010} one must
compute the number of possible draws in a $4 \times 4 \times 4$ game of
Tic-Tac-Toe, where only $N$ crosses are placed and all other spaces are filled
with naughts. This amounts to counting the number of satisfying assignments of
the following formula.
\begin{equation*}
  \mathit{init}(N) \wedge \bigwedge_{(i,j,k,l) \in L}
  \left(
    (x_i \vee x_j \vee x_k \vee x_l)
    \wedge
    (\overline{x_i} \vee \overline{x_j} \vee \overline{x_k} \vee \overline{x_l})
  \right)
  \enspace ,
\end{equation*}
where $\mathit{init}(N)$ is true iff exactly $N$ out of the $76$ variables are
true, and $L$ is the set of all $76$ lines through the $4 \times 4 \times 4$
cube. To minimise the time and space to solve, lines in $L$ are accumulated in
increasing order of the number of levels they span. The ROBDDs for both
$\mathit{init}(N)$ and the $76$ line formulas can be directly constructed
without any BDD operations. Hence, this benchmark always consists of $76$ uses
of Apply to accumulate the line constraints onto $\mathit{init}(N)$.

\paragraph{Picotrav. }
The \emph{EPFL} Combinational Benchmark Suite~\cite{Amaru2015} consists of 23
combinatorial circuits designed for logic optimisation and synthesis. 20 of
these are split into the two categories \emph{random/control} and
\emph{arithmetic}, and each of these original circuits $C_o$ are distributed
together with one optimised for size $C_s$ and one for optimised for depth
$C_d$. The last three benchmarks are the \emph{More than a Million Gates}
benchmarks, which we will ignore as they come without optimised versions.
They also are generated randomly and hence they are unrealistic anyway.

We have recreated a subset of the \emph{Nanotrav} BDD application as distributed
with CUDD~\cite{Somenzi2015}. Here, we verify the functional equivalence between
each output gate of the original circuit $C_o$ and the corresponding gate of an
optimised circuit $C_d$ or $C_s$. Every input gate is represented by a decision
variable and recursively the BDD representing each gate is computed. Memoisation
ensures the same gate is not computed twice, while a reference counter is
maintained for each gate to clear the memoisation table; this allows for garbage
collection of intermediate BDDs. Finally, every pair of BDDs that should
represent the same output are tested for equality.

The variable order used was chosen based on what produced the smallest BDDs
during our initial experiments. We had the \emph{random/control} benchmarks
use the order in which the inputs were declared while the \emph{arithmetic}
benchmarks derived an ordering based on the deepest reference within the
optimised circuit to their respective input gate; ties for the same level are
resolved by a DFS traversal of the same circuit.

\subsection{Hardware}

We have run experiments on the following two very different kinds of machines.
\begin{itemize}
\item \emph{Consumer grade laptop} with one quadro-core 2.6 GHz Intel i7-4720HQ
  processor, 8 GiB of RAM, 230 GiB of available SSD disk, running Ubuntu, and
  compiling code with GCC 9.3.0.

\item \emph{Grendel server node} with two 48-core 3.0 GHz Intel Xeon Gold 6248R
  processors, 384 GiB of RAM, 3.5 TiB of available SSD disk, running CentOS
  Linux, and compiling code with GCC 10.1.0.
\end{itemize}
The former, due to its limited RAM, has until now been incapable of manipulating
larger BDDs. The latter, on the other hand, has enough RAM and disk to allow all
BDD implementations to solve large instances on comparable hardware.

\subsection{Experimental Results}

All but the largest benchmarks were run multiple times and the \emph{minimum}
measured running time is reported, since it minimises any error caused by
slowdown and overhead on the CPU and the memory \cite{Chen2016}. Using the
\emph{average} or \emph{median} value instead has only a negligible impact on
the resulting numbers.

\subsubsection{Research Question \ref{rq:large instances}:} \label{sec:experiments:large instances}

Fig.~\ref{fig:rq_1:var_memory} shows the $15$-Queens problem being solved with
\Adiar\ on the Grendel server node with variable amounts of available main
memory. \emph{cgroups} are used to enforce the machines memory limit, including
its file system cache, to be only a single GiB more than what is given to
\Adiar. During these computations, the largest reduced BDD is $19.4$~GiB which
makes its unreduced input at least $25.9$~GiB in size. That is, the input and
output of the largest run of Reduce alone occupies at least $45.3$~GiB. Yet, we
only see a $39.1\%$ performance decrease when decreasing the available memory
from $256$~GiB to $2$~GiB. This change in performance primarily occurs in the
$2$~GiB to $64$~GiB interval where data needs to be fetched from disk; more than
half of this decrease is in the $2$~GiB to $12$~GiB interval.

\begin{figure}[ht!]
  \centering

  \subfloat[Computing $15$-Queens with different amount of memory available.]{
    \label{fig:rq_1:var_memory}

    \begin{tikzpicture}
      \begin{axis}[%
        width=0.46\linewidth, height=0.34\linewidth,
        % x-axis
        xlabel={memory (GiB)},
        xmode=log,
        log basis x = 2,
        xmajorgrids=true,
        xminorticks=false,
        xticklabels={2,8,32,128},
        extra x ticks={4,16,64,256},
        extra x tick label={\empty},
        % y-axis
        ymin=3000,
        axis y discontinuity=crunch, %other options: parallel
        ytickmin=4000,
        ylabel={time (seconds)},
        % extra y ticks={4100,4300,4500},
        extra y tick labels={,,},
        yminorgrids=true,
        ymajorgrids=true,
        grid style={dashed,black!10}
        ]
        \addplot+ [style=plot_adiar]
        table {./data/memory_15_queens_adiar_time.tex};
      \end{axis}
    \end{tikzpicture}
  }
  \quad
  \subfloat[Computing $N$-Queens solutions with $7$ GiB of memory available.]{
    \label{fig:rq_1:var_n}
    \begin{tikzpicture}
      \begin{axis}[%
        width=0.46\linewidth, height=0.34\linewidth,
        % x-axis
        xlabel={N},
        xticklabels={\empty},
        extra x ticks={7,9,11,13,15},
        xmajorgrids=true,
        % y-axis
        ylabel={$\mu$s / BDD node},
%        ytick distance={10},
%        ymode=log,
        yminorgrids=false,
        ymajorgrids=true,
        grid style={dashed,black!10}
        ]
        \addplot+ [style=plot_adiar]
        table {./data/memory_N_queens_adiar_time_per_node.tex};
      \end{axis}
    \end{tikzpicture}
  }
  
  \caption{Performance of \Adiar\ in relation to available memory.}
  \label{fig:rq_1}
\end{figure}

To confirm the seamless transition to disk, we investigate different $N$, fix
the memory, and normalise the data to be $\mu$s per node. The current version of
\Adiar\ is implemented purely using the external memory algorithms of \TPIE.
These perform poorly when given small amounts of data; the time it takes to
initialise the larger memory makes it by several orders of magnitude slower than
if it used equivalent internal memory algorithms. This overhead is apparent for
$N \leq 11$.

The consumer grade laptop's memory cannot contain the $19.4$~GiB BDD and its SSD
the total of $250.3$ GiB of data generated by the $15$-Queens problem. Yet, as
Fig.~\ref{fig:rq_1:var_n} shows for $N = 7, \dots, 15$ the computation time per
node only slows down by a factor of $1.8$ when crossing the memory barrier from
$N = 14$ to $15$. Furthermore, this is only a slowdown at $N=15$ by a factor of
$2.02$ compared to the lowest recorded computation time per node at $N=12$.

\subsubsection{Research Question \ref{rq:priority queue}:} \label{sec:experiments:priority queue}

Table~\ref{tab:priority_queue} shows the average running time of the $N$-Queens
problem for $N=14$, resp.\ the Tic-Tac-Toe problem for $N=22$, when using the
levelized priority queue compared to the priority queue of TPIE.

\begin{table}[ht!]
  \centering

  \subfloat[N-Queens ($N = 14$)]{
    \begin{tabular}{ c || c }
      Priority Queue & Time (s)
      \\ \hline \hline
      TPIE & $908$
      \\
      L-PQ($1$) & $678$
      \\
      L-PQ($4$) & $677$
    \end{tabular}
  }
  \subfloat[Tic-Tac-Toe ($N = 22$)]{
    \begin{tabular}{ c || c }
      Priority Queue & Time (s)
      \\ \hline \hline
      TPIE & $1003$
      \\
      L-PQ($1$) & $632$
      \\
      L-PQ($4$) & $675$
    \end{tabular}
  }
  
  \caption{Performance increase by use of the levelized priority queue with $k$
    buckets (L-PQ($k$)) compared to the priority queue of TPIE.}
  \label{tab:priority_queue}
\end{table}

Performance increases by $25.3\%$ for the Queens and by $37.0\%$ for the
Tic-Tac-Toe benchmark when switching from the TPIE priority queue to the
levelized priority queue with a single bucket. The BDDs generated in either
benchmark have very few (if any) arcs going across more than a single level,
which explains the lack of any performance increase past the first bucket.

\subsubsection{Research Question \ref{rq:leaf stream}:} \label{sec:experiments:leaf stream}

To answer this question, we move the contents of $\Lredleaf$ into $\Qred$ at the
start of the Reduce algorithm. This is not possible with the levelized priority
queue, so the experiment is conducted on the consumer grade laptop with $7$ GiB
of ram dedicated to an older version of \Adiar\ with the priority queue of TPIE.
The node-to-leaf arcs make up $47.5\%$ percent of all arcs generated in the
$14$-Queens benchmark. The average time to solve this goes down from $1258$ to
$919$ seconds. This is an improvement of $27.0\%$ which is $0.57\%$ for every
percentile the node-to-leaf arcs contribute to the total number of arcs
generated. Compensating for the performance increase in Research
Question~\ref{rq:priority queue} this only amounts to $20.12\%$, i.e.\ $0.42\%$
per percentile.

\begin{table}[p]
  \centering

  \small

  \subfloat[N-Queens]{
    \label{tab:pruning:queens}
    
    \begin{tabular}{ c | c | c | c }
      \multicolumn{4}{c}{N}
      \\
      $12$ & $13$ & $14$ & $15$
      \\ \hline \hline
      \multicolumn{4}{c}{Time to solve (s)}
      \\
      $3.70 \cdot 10^1$ & $2.09 \cdot 10^2$ & $1.35 \cdot 10^3$ & $1.21 \cdot 10^4$
      \\
      $2.25 \cdot 10^1$ & $1.23 \cdot 10^2$ & $7.38 \cdot 10^2$ & $6.34 \cdot 10^3$
      \\ \hline \hline
      \multicolumn{4}{c}{Largest unreduced size (\#nodes)}
      \\
      $9.97 \cdot 10^6$ & $5.26 \cdot 10^7$ & $2.76 \cdot 10^8$ & $1.70 \cdot 10^9$
      \\
      $7.33 \cdot 10^6$ & $3.93 \cdot 10^7$ & $2.10 \cdot 10^8$ & $1.29 \cdot 10^9$
      \\ \hline \hline
      \multicolumn{4}{c}{Median unreduced size (\#nodes)}
      \\
      $2.47 \cdot 10^3$ & $3.92 \cdot 10^3$ & $6.52 \cdot 10^3$ & $1.00 \cdot 10^4$
      \\
      $2.47 \cdot 10^3$ & $3.92 \cdot 10^3$ & $6.52 \cdot 10^3$ & $1.00 \cdot 10^4$
      \\ \hline \hline
      \multicolumn{4}{c}{Node-to-leaf arcs ratio}
      \\
      $16.9 \%$ & $17.4 \%$ & $17.6 \%$ & $17.4 \%$
      \\
      $43.9 \%$ & $46.2 \% $ & $47.5 \%$ & $48.1 \%$
    \end{tabular}
  }
  
  \subfloat[Tic-Tac-Toe]{
    \label{tab:pruning:tic_tac_toe}
    
    \begin{tabular}{ c | c | c | c }
      \multicolumn{4}{c}{N}
      \\
      $20$ & $21$ & $22$ & $23$
      \\ \hline \hline
      \multicolumn{4}{c}{Time to solve (s)}
      \\
      $2.80 \cdot 10^1$ & $1.57 \cdot 10^2$ & $7.99 \cdot 10^2$ & $1.20 \cdot 10^4$
      \\
      $2.38 \cdot 10^1$ & $1.39 \cdot 10^2$ & $6.96 \cdot 10^2$ & $8.91 \cdot 10^3$
      \\ \hline \hline
      \multicolumn{4}{c}{Largest unreduced size (\# nodes)}
      \\
      $2.44 \cdot 10^6$ & $1.26 \cdot 10^7$ & $5.97 \cdot 10^7$ & $2.59 \cdot 10^8$
      \\
      $2.44 \cdot 10^6$ & $1.26 \cdot 10^7$ & $5.97 \cdot 10^7$ & $2.59 \cdot 10^8$
      \\ \hline \hline
      \multicolumn{4}{c}{Median size (\# nodes)}
      \\
      $1.13 \cdot 10^4$ & $1.52 \cdot 10^4$ & $1.87 \cdot 10^4$ & $2.18 \cdot 10^4$
      \\
      $8.79 \cdot 10^3$ & $1.19 \cdot 10^4$ & $1.47 \cdot 10^4$ & $1.73 \cdot 10^4$
      \\ \hline \hline
      \multicolumn{4}{c}{Node-to-leaf arcs ratio}
      \\
      $8.73 \%$ & $7.65 \%$ & $6.67 \%$ & $5.80 \%$
      \\
      $22.36 \%$ & $18.34 \%$ & $14.94 \%$ & $12.29 \%$
    \end{tabular}
  }
  
  \caption{Effect of pruning on performance. The first row for each feature is
    \emph{without} pruning and the second is \emph{with} pruning.}
  \label{tab:pruning}
\end{table}

\subsubsection{Research Question \ref{rq:pruning}:} \label{sec:experiments:pruning}

Like in Research Question~\ref{rq:leaf stream}, it is to be expected that this
optimisation is dependant on the input. Table~\ref{tab:pruning} shows different
benchmarks run on the consumer grade laptop with and without pruning.

For $N$-Queens the largest unreduced BDD decreases in size by at least $23\%$
while the median is unaffected. The $x_{ij} \land \neg
\mathit{has\_threat}(i,j)$ base case consists mostly of $\bot$ leaves, so they
can prune the outermost conjunction of rows but not the disjunction within each
row. The relative number of node-to-leaf arcs is also at least doubled, which
decreases the number of elements placed in the priority queues. This, together
with the decrease in the largest unreduced BDD, explains how the computation
time decreases by $49\%$ for $N = 15$. Considering our results for Research
Question~\ref{rq:leaf stream} above at most half of that speedup can be
attributed to increasing the percentage of node-to-leaf arcs.

We observe the opposite for the Tic-Tac-Toe benchmark. The largest unreduced
size is unaffected while the median size decreases by at least $20\%$. Here, the
BDD for each line has only two arcs to the $\bot$ leaf that can short circuit the
accumulated conjunction, while the largest unreduced BDDs are created when
accumulating the very last few lines, that span almost all levels. Still, there
is a doubling in the total ratio of node-to-leaf arcs and we see at least
an $11.6\%$ decrease in the computation time.

\subsubsection{Research Question \ref{rq:equality}:} \label{sec:experiments:equality}

Table~\ref{tab:equality} shows the time to do equality checking on the three
largest circuits verified with the Picotrav application run on the Grendel
server nodes using \Adiar\ with $300$~GiB available. The number of nodes and
size reported in the table is of a single set of BDDs that describe the final
circuit. That is, since \Adiar\ does not share nodes, then the set of final BDDs
for the specification and the optimised circuit are distinct; the
\emph{mem\_ctrl} benchmark requires 1231 isomorphism checks on a total of $2
\cdot 2.68 \cdot 10^{10}$ nodes which is equivalent to $1.17$~TiB of data. In
these benchmarks all equality checking was possible to do with a weighted
average performance of $0.109$~$\mu$s/node.

\if\arxiv1
\begin{table}[t]
  \else
\begin{table}[t]
  \fi
  \centering

  \subfloat[\emph{mem\_ctrl}]{    
    \begin{tabular}{ l || c c }
      & depth & size
      \\ \hline \hline      
      Time (s)        & 5862    & 5868
      \\ %\hline
      $O(\sort(N))$   & 496     & 476
      \\
      $O(N/B)$        & 735     & 755
      \\ \hline
      \if\arxiv1
      $N$ (BDD nodes) & \multicolumn{2}{c}{$2.68 \cdot 10^{10}$}
      \\ \fi
      $N$ (MiB)       & \multicolumn{2}{c}{614313}
    \end{tabular}
  }
  \subfloat[\emph{sin}]{
    \begin{tabular}{ l || c c }
      & depth & size
      \\ \hline \hline      
      Time (s)         & 3.89 & 3.27
      \\ %\hline
      $O(\sort(N))$    & 22    & 22
      \\
      $O(N/B)$         & 3     & 3
      \\ \hline
      \if\arxiv1
      $N$ (BDD nodes)  & \multicolumn{2}{c}{$8.02 \cdot 10^6$}
      \\ \fi
      $N$ (MiB)        & \multicolumn{2}{c}{3589}
    \end{tabular}
  } %
  \if\arxiv1 %
    \\
  \fi %
  \subfloat[\emph{voter}]{
    \begin{tabular}{ l || c c }
      & depth & size
      \\ \hline \hline
      Time (s)        & 0.058 & 0.006
      \\ %\hline
      $O(\sort(N))$   & 1     & 0
      \\
      $O(N/B)$        & 0     & 1
      \\ \hline
      \if\arxiv1
      $N$ (BDD nodes) & \multicolumn{2}{c}{$2.51 \cdot 10^5$}
      \\ \fi
      $N$ (MiB)       & \multicolumn{2}{c}{5.74}
    \end{tabular}
  }
  
  \caption{Running time for equivalence testing. $O(\sort(N))$ and $O(N/B)$
    is the number of times the respective algorithm in
    Section~\ref{sec:theory:new algorithms:equality} was used.}
  \label{tab:equality}
\end{table}

The \emph{voter} benchmark is especially interesting, since it consists only of
a single output gate and the $O(\sort(N))$ and $O(N/B)$ algorithm are used
respectively in the depth and size optimised instance. As witnessed in
Section~\ref{sec:experiments:large instances}, \Adiar\ has a bad performance for
smaller instances. Yet, despite only a total of $11.48$~MiB of data is compared,
the $O(\sort(N))$ algorithm runs at $0.116$~$\mu$s/node and the $O(N/B)$ scan at
only $0.012$~$\mu$s/node. That is, the $O(N/B)$ algorithm can compare at least
${2 \cdot 5.75~\text{MiB}} / {0.006~\text{s}} = 1.86$~GiB/s.

Table~\ref{tab:equality_slow} shows the running time of equality checking by
instead computing whether $f \leftrightarrow g = \top$. Not even taking the speedup
due to the priority queue and separation of node-to-leaf arcs in
Section~\ref{sec:experiments:priority queue} and \ref{sec:experiments:leaf
  stream} into account, this approach, as is necessary with Arge's original
algorithms, is $2.42$ -- $63.5$ slower than using our improved algorithms.

\begin{table}[ht!]
  \centering

  \subfloat[\emph{mem\_ctrl}]{
    \begin{tabular}{ l || c c }
      & depth & size
      \\ \hline \hline      
      Time (s) & 44233.2 & 44263.4 
    \end{tabular}
  }
  \subfloat[\emph{sin}]{
    \begin{tabular}{ l || c c }
      & depth & size
      \\ \hline \hline      
      Time (s) & 10.165 & 7.911 
    \end{tabular}    
  } %
  \if\arxiv1 %
    \\
  \fi %
  \subfloat[\emph{voter}]{
    \begin{tabular}{ l || c c }
      & depth & size
      \\ \hline \hline      
      Time (s) & 0.380 & 0.381
    \end{tabular}    
  }

  \caption{Running time for checking equality with $f \leftrightarrow g = \top$}
  \label{tab:equality_slow}
\end{table}

Out of the $3861$ output gates checked for equality throughout the $23$ verified
circuits the $O(N/B)$ linear scan could be used for $71.6\%$ of them.

\subsubsection{Research Question \ref{rq:performance comparison}:}

We have compared the performance of \emph{\mbox{\Adiar}}~$1.0.1$ with the
\emph{\mbox{BuDDy}}~$2.4$~\cite{Lind1999}, the
\emph{\mbox{CUDD}}~$3.0.0$~\cite{Somenzi2015}, and the
\emph{\mbox{Sylvan}}~$1.5.0$~\cite{Dijk2016} BDD package.

To this end, we ran all our benchmarks on Grendel server nodes, which were set
to use $350$ GiB of the available RAM, while each BDD package is given $300$~GiB
of it. Sylvan was set to not use any parallelisation and given a ratio between
the unique node table and the computation cache of $64{:}1$. BuDDy was set to
the same cache-ratio and the size of the CUDD cache was set such it would have
an equivalent ratio when reaching its $300$~GiB limit. The I/O analysis in
Section~\ref{sec:preliminaries:bdd:io} is evident in the running time of
Sylvan's implicit Reduce, which increases linearly with the size of the node
table\footnote{Experiments using \emph{perf} on Sylvan show that dereferencing a
  bucket in the unique node table and using x86 locks to obtain exclusive
  ownership of cache lines are one of the main reasons for the slowdown. We
  hypothesise that the overhead of \emph{mmap} is the main culprit.}. Hence,
Sylvan has been set to start its table $2^{12}$ times smaller than the final
$262$~GiB it may occupy, i.e.\ at first with a table and cache that occupies
$66$~MiB.

As is evident in Section~\ref{sec:experiments:large instances}, the slowdown for
\Adiar\ for small instances, due to its use of external memory algorithms, makes
it meaningless to compare its performance to other BDD packages when the largest
BDD is smaller than $32$~MiB. Hence, we have chosen to omit these instances from
this report, though the full data set (including these instances) is publicly
available.

\paragraph{N-Queens. }

Fig.~\ref{fig:queens:n} shows for each BDD package their running time computing
the $N$-Queens benchmark for $12 \leq N \leq 17$. At $N=13$, where \Adiar's
computation time per node is the lowest, \Adiar\ runs by a factor of $5.1$
slower than BuDDy, $2.3$ than CUDD, and $2.6$ than Sylvan. The gap in
performance of \Adiar\ and other packages gets smaller as instances grow larger:
for $N=15$, which is the largest instance solvable by CUDD and Sylvan, \Adiar\
is only slower than CUDD, resp. Sylvan, by a factor of $1.47$, resp. $2.15$.

\if\arxiv1
\begin{figure}[ht!]
  \else
\begin{figure}[t]
  \fi
  \centering

  \subfloat{
    \label{fig:queens:time}
    
    \begin{tikzpicture}
      \begin{axis}[%
        width=0.48\linewidth, height=0.34\linewidth,
        every tick label/.append style={font=\scriptsize},
        % x-axis
        xlabel={N},
        xmajorgrids=true,
        xtick={11,...,17},
        % y-axis
        ylabel={s},
        ytick distance={10},
        ymode=log,
        yminorgrids=false,
        ymajorgrids=true,
        grid style={dashed,black!12},
        ]
        \addplot+ [style=plot_adiar]
        table {./data/queens_adiar_time.tex};

        \if\arxiv1
        \addplot+ [style=plot_buddy]
        table {./data/queens_buddy_time.tex};
        \fi

        \addplot+ [style=plot_cudd]
        table {./data/queens_cudd_time.tex};

        \addplot+ [style=plot_sylvan]
        table {./data/queens_sylvan_time.tex};
      \end{axis}
    \end{tikzpicture}
  }
  \quad
  \subfloat{
    \label{fig:queens:time per node}
    
    \begin{tikzpicture}
      \begin{axis}[%
        width=0.48\linewidth, height=0.34\linewidth,
        every tick label/.append style={font=\scriptsize},
        % x-axis
        xlabel={N},
        xmajorgrids=true,
        xtick={11,...,17},
        % y-axis
        ylabel={$\mu$s / BDD node},
        yminorgrids=true,
        ymajorgrids=true,
        ymax=2.2,
        grid style={dashed,black!20},
        ]
        \addplot+ [style=plot_adiar]
        table {./data/queens_adiar_time_per_node.tex};

        \if\arxiv1
        \addplot+ [style=plot_buddy]
        table {./data/queens_buddy_time_per_node.tex};
        \fi

        \addplot+ [style=plot_cudd]
        table {./data/queens_cudd_time_per_node.tex};

        \addplot+ [style=plot_sylvan]
        table {./data/queens_sylvan_time_per_node.tex};
      \end{axis}

    \end{tikzpicture}
  }

  \begin{tikzpicture}
    \if\arxiv1
    \begin{customlegend}[
      legend columns=-1,
      legend style={draw=none,column sep=1ex},
      legend entries={\Adiar, BuDDy, CUDD, Sylvan}
      ]
      \addlegendimage{style=plot_adiar}
      \addlegendimage{style=plot_buddy}

      \addlegendimage{style=plot_cudd}
      \addlegendimage{style=plot_sylvan}
    \end{customlegend}
    \else
    \begin{customlegend}[
      legend columns=-1,
      legend style={draw=none,column sep=1ex},
      legend entries={\Adiar, CUDD, Sylvan}
      ]
      \addlegendimage{style=plot_adiar}
      \addlegendimage{style=plot_cudd}
      \addlegendimage{style=plot_sylvan}
    \end{customlegend}
    \fi
  \end{tikzpicture}
  
  \caption{Running time solving $N$-Queens (lower is better).}
  \label{fig:queens:n}
\end{figure}

\Adiar\ outperforms all three libraries in terms of successfully computing very
large instances. The largest BDD constructed by \Adiar\ for $N=17$ is $719$~GiB
in size, whereas Sylvan with $N=15$ only constructs BDDs up to $12.9$~GiB in
size. Yet, at $N=17$, where \Adiar\ has to make heavy use of the disk, \Adiar's
computation time per node only slows down by a factor of $1.8$ compared to its
performance at $N=13$.

\if\arxiv1
\begin{figure}[t]
  \else
\begin{figure}[b]
  \fi
  \centering

  \subfloat{
    \begin{tikzpicture}
      \begin{axis}[%
        width=0.8\linewidth, height=0.34\linewidth,
        every tick label/.append style={font=\scriptsize},
        % x-axis
        xlabel={Memory (GiB)},
        xmin={-10},
        xmax={265},
        xmajorgrids=true,
        % y-axis
        ylabel={s},
        ymin={0},
        ymax={6000},
        ytick distance={1000},
        ymajorgrids=true,
        grid style={dashed,black!12},
        ]
        \addplot+ [style=plot_adiar]
        table {./data/memory_15_queens_adiar_time.tex};

        \addplot+ [style=plot_cudd]
        table {./data/memory_15_queens_cudd_time.tex};

        \addplot+ [style=plot_sylvan]
        table {./data/memory_15_queens_sylvan_time.tex};
      \end{axis}
    \end{tikzpicture}    
  }
  \subfloat{
    \begin{tikzpicture}
      \begin{customlegend}[
        legend columns=1,
        legend style={draw=none,column sep=1ex},
        legend entries={\Adiar, CUDD, Sylvan}
        ]
        \addlegendimage{style=plot_adiar}
        \addlegendimage{style=plot_cudd}
        \addlegendimage{style=plot_sylvan}
      \end{customlegend}
    \end{tikzpicture}    
  }
  
  \caption{Running time of $15$-Queens with variable memory (lower is better).}
  \label{fig:queens:mem}
\end{figure}

Conversely, \Adiar\ is also able to solve smaller instances with much less
memory than other packages. Fig~\ref{fig:queens:mem} shows the running time for
both \Adiar, CUDD, and Sylvan solving the $15$-Queens problem depending on the
available memory. Sylvan was first able to solve this problem when given
$56$~GiB of memory while CUDD, presumably due to its larger node-size and
multiple data structures, requires $72$~GiB of memory to be able to compute the
solution.

\paragraph{Tic-Tac-Toe. }

The running times we obtained for this benchmark, as shown in
Fig.~\ref{fig:tic_tac_toe}, paint the same picture as for Queens above: the
factor with which Adiar\ runs slower than the other packages decreases as the
size of the instance increases. At $N=24$, which is the largest instance solved
by CUDD, \Adiar\ runs slower than CUDD by a factor of $2.38$. The largest
instance solved by Sylvan is $N=25$ where the largest BDD created by Sylvan is
$34.4$~GiB in size, and one incurs a $2.50$ factor slowdown by using \Adiar\
instead.

\begin{figure}[ht!]
  \centering

  \subfloat[]{
    \label{fig:tic_tac_toe:time}
    
    \begin{tikzpicture}
      \begin{axis}[%
        width=0.48\linewidth, height=0.34\linewidth,
        every tick label/.append style={font=\scriptsize},
        % x-axis
        xlabel={N},
        xmajorgrids=true,
        xtick={20,...,29},
        % y-axis
        ylabel={s},
        ytick distance={10},
        ymode=log,
        yminorgrids=false,
        ymajorgrids=true,
        grid style={dashed,black!12},
        ]
        \addplot+ [style=plot_adiar]
        table {./data/tic_tac_toe_adiar_time.tex};

        \if\arxiv1
        \addplot+ [style=plot_buddy]
        table {./data/tic_tac_toe_buddy_time.tex};
        \fi

        \addplot+ [style=plot_cudd]
        table {./data/tic_tac_toe_cudd_time.tex};

        \addplot+ [style=plot_sylvan]
        table {./data/tic_tac_toe_sylvan_time.tex};
      \end{axis}
    \end{tikzpicture}
  }
  \subfloat[]{
    \label{fig:tic_tac_toe:time per node}
    
    \begin{tikzpicture}
      \begin{axis}[%
        width=0.48\linewidth, height=0.34\linewidth,
        every tick label/.append style={font=\scriptsize},
        % x-axis
        xlabel={N},
        xmajorgrids=true,
        xtick={20,...,29},
        % y-axis
        ylabel={$\mu$s / BDD node},
        yminorgrids=true,
        ymajorgrids=true,
        ymax=3,
        grid style={dashed,black!20},
        ]
        \addplot+ [style=plot_adiar]
        table {./data/tic_tac_toe_adiar_time_per_node.tex};

        \if\arxiv1
        \addplot+ [style=plot_buddy]
        table {./data/tic_tac_toe_buddy_time_per_node.tex};
        \fi

        \addplot+ [style=plot_cudd]
        table {./data/tic_tac_toe_cudd_time_per_node.tex};

        \addplot+ [style=plot_sylvan]
        table {./data/tic_tac_toe_sylvan_time_per_node.tex};
      \end{axis}

    \end{tikzpicture}
  }

  \begin{tikzpicture}
    \if\arxiv1
    \begin{customlegend}[
      legend columns=-1,
      legend style={draw=none,column sep=1ex},
      legend entries={\Adiar, BuDDy, CUDD, Sylvan}
      ]
      \addlegendimage{style=plot_adiar}
      \addlegendimage{style=plot_buddy}

      \addlegendimage{style=plot_cudd}
      \addlegendimage{style=plot_sylvan}
    \end{customlegend}
    \else
    \begin{customlegend}[
      legend columns=-1,
      legend style={draw=none,column sep=1ex},
      legend entries={\Adiar, CUDD, Sylvan}
      ]
      \addlegendimage{style=plot_adiar}
      \addlegendimage{style=plot_cudd}
      \addlegendimage{style=plot_sylvan}
    \end{customlegend}
    \fi
  \end{tikzpicture}
  
  \caption{Running time solving Tic-Tac-Toe (lower is better).}
  \label{fig:tic_tac_toe}
\end{figure}

\Adiar\ was able to solve the instance of $N=29$, where the largest BDD created
was $902$~GiB in size. Yet, even though the disk was extensively used, \Adiar's
computation time per node only slows down by a factor of $2.49$ compared to its
performance at $N=22$.

\paragraph{Picotrav. }

Table~\ref{tab:picotrav} shows the number of successfully verified circuits by
each BDD package and the number of benchmarks that were unsuccessful due to a
full node table or a full disk and the ones that timed out after $15$ days of
computation time.

\if\arxiv1
\begin{table}[ht!]
  \else
\begin{table}[b]
  \fi
  \centering

  \begin{tabular}{ l || c | c | c}
                    & \# solved & \# out-of-space & \# time-out
    \\ \hline \hline
    \Adiar       & 23        & 6          & 11
    \\
    \if\arxiv1
    BuDDy        & 19        & 20         & 1
    \\
    \fi    
    CUDD         & 20        & 19         & 1
    \\
    Sylvan       & 20        & 13         & 7
    \end{tabular}

    \medskip

    \caption{Number of verified \emph{arithmetic} and \emph{random/control}
      circuits from \cite{Amaru2015}}
  \label{tab:picotrav}
\end{table}

If Sylvan's unique node table and computation cache are immediately instantiated
to their full size of $262$~GiB then it is able to verify $3$ more of the $40$
circuits within the $15$ days time limit. One of these three is the
\emph{arbiter} benchmark optimised with respect to size that BuDDy and CUDD are
able to solve in a few seconds. Yet, BuDDy also exhibits a similar slowdown when
it has to double its unique node table to its full size. We hypothesise this
slowdown is due to the computation cache being cleared when nodes are rehashed
into the doubled node table, while this benchmark consists of a lot of repeated
computations. The other are the \emph{mem\_ctrl} benchmark optimised with
respect to size and depth.

The performance of \Adiar\ compared to the other packages is reminiscent to our
results from the two other benchmarks. For example, the \emph{voter} benchmark,
where the largest BDD for \Adiar\ is $8.2$~GiB in size, it is $3.69$ times
slower than CUDD and $3.07$ times slower than Sylvan. In the \emph{mem\_ctrl}
benchmark optimised for depth which Sylvan barely was able to solve by changing
its cache ratio, \Adiar\ is able to construct the BDDs necessary for the
comparison $1.01$ times faster than Sylvan. This is presumably due to the large
overhead of Sylvan in having to repeatedly run its garbage collection
algorithms.

Despite the fact that the disk available to \Adiar\ is only $12$ times larger
than internal memory, \Adiar\ has to explicitly store both the unreduced and
reduced BDDs, and many of the benchmarks have hundreds of BDDs concurrently
alive, \Adiar\ is still able to solve the same benchmarks as the other packages.
For example, the largest solved benchmark, \emph{mem\_ctrl}, has at one point
$1231$ different BDDs in use at the same time.

\section{Conclusions and Future Work} \label{sec:conclusion}

We propose I/O-efficient BDD manipulation algorithms in order to scale BDDs
beyond main memory. These new iterative BDD algorithms exploit a total
topological sorting of the BDD nodes, by using priority queues and sorting
algorithms. All recursion requests for a single node are processed at the same
time, which fully circumvents the need for a memoisation table. If the
underlying data structures and algorithms are cache-aware or cache-oblivious,
then so are our algorithms by extension. The I/O complexity of our time-forward
processing algorithms is compared to the conventional recursive algorithms on a
unique node table with complement edges~\cite{Brace1990} in
Table~\ref{tab:io_comparison}.

The performance of \Adiar\ is very promising in practice for instances larger
than a few hundred MiB. As the size of the BDDs increase, the performance of
\Adiar\ gets closer to conventional recursive BDD implementations. When the BDDs
involved in the computation exceeds a few GiB then the use of \Adiar\ only
results in a $3.69$ factor slowdown compared to Sylvan and CUDD -- it was only
$1.47$ times slower than CUDD in the largest Queens benchmark that CUDD could
solve. Simultaneously, the design of our algorithms allow us to compute on BDDs
that outgrow main memory with only a $2.49$ factor slowdown, which is negligible
in comparison to the slowdown that conventional BDD packages experience when using
swap memory.

This performance comes at the cost of not sharing any nodes and so not being
able to compare for functional equivalence in $O(1)$ time. We have improved the
asymptotic behaviour of equality checking to only be an $O(\sort(N))$ algorithm
which in practice is negligible compared to the time to construct the BDDs
involved. For $71.6\%$ of all the output gates we verified from \emph{EPFL}
Combinational Benchmark Suite~\cite{Amaru2015} we were even able to do so with a
simple $O(N/B)$ linear scan that can compare more than $1.86$~GiB/s. This number
is realistic, since modern SSDs, depending on block size used, have a sequential
transfer rate of $1$~GiB/s and $2.8$~GiB/s.

In practice, the fact that nodes are not shared does not negatively impact the
ability of \Adiar\ to solve a problem in comparison to conventional BDD
packages. This is despite the ratio between disk and main memory is smaller than
the number of BDDs in use. Furthermore, garbage collection becomes a trivial and
cheap deletion of files on disk.

This lays the foundation on which we intend to develop external memory BDD
algorithms usable in the fixpoint algorithms for symbolic model checking. We
will, to this end, attempt to improve further on the non-constant equality
checking, improve the performance of quantifying multiple variables, and design
an I/O-efficient relational product function. Furthermore, we intend to improve
\Adiar's performance for smaller instances by processing them exclusively in
internal memory and generalise its implementation to also include
Multi-Terminal~\cite{Fujita1997} and Zero-suppressed~\cite{Minato1993} Decision
Diagrams.

\begin{table}[ht!]
  \centering
  \begin{tabular}{r l || c | c}
    \multicolumn{2}{c||}{Algorithm} & Depth-first & Time-forwarded
    \\ \hline \hline
    Reduce         &                                    & $O(N)$                       & $O(\sort(N))$
    \\ \hline \multicolumn{4}{c}{BDD Manipulation} \\ \hline
    Apply          & $f \odot g$                        & $O(N_f N_g)$                 & $O(\sort(N_f N_g))$
    \\
    If-Then-Else   & $f\ ?\ g\ :\ h$                    & $O(N_f N_g N_h)$ & $O(\sort (N_f N_g N_h))$
    \\
    Restrict       & $f|_{x_i = v}$                      & $O(N_f)$                     & $O(\sort(N_f))$
    \\
    Negation       & $\neg f$                           & $O(1)$                       & $O(1)$
    \\
    Quantification & $\exists/\forall v : f|_{x_i = v}$  & $O(N_f^2)$                   & $O(\sort(N_f^2))$
    \\
    Composition    & $f|_{x_i = g}$                      & $O(N_f^2 N_g)$               & $O(\sort(N_f^2 N_g))$
    \\ \hline \multicolumn{4}{c}{Counting} \\ \hline
    Count Paths    & $\# $paths in $f$ to $\top$        & $O(N_f)$                     & $O(\sort(N_f))$
    \\
    Count SAT      & $\# x : f(x)$                      & $O(N_f)$                     & $O(\sort(N_f))$
    \\ \hline \multicolumn{4}{c}{Other} \\ \hline
    Equality       & $f \equiv g$                       & $O(1)$                       & $O(\sort(\min(N_f,N_g)))$
    \\
    Evaluate       & $f(x)$                             & $O(L_f)$                     & $O(N_f/B)$
    \\
    Min SAT    & $\min\{ x \mid f(x) \}$                & $O(L_f)$                     & $O(N_f/B)$
    \\
    Max SAT    & $\max\{ x \mid f(x) \}$                & $O(L_f)$                     & $O(N_f/B)$
  \end{tabular}
  
  \caption{I/O-complexity of conventional depth-first algorithms compared to the
    time-forwarded we propose. Here, $N/B < \sort(N) \triangleq N/B \cdot
    \log_{M/B}(N/B) \ll N$, where $N$ is the number of nodes, and $L$ the number
    of levels in the BDD.}
  \label{tab:io_comparison}
\end{table}

\section*{Acknowledgements}
Thanks to the late Lars Arge for his input and advice and to Mathias Rav for
helping us with \TPIE\ and his inputs on the levelized priority queue. Thanks to
Gerth S. Brodal, Asger H. Drewsen for their help, Casper Rysgaard for his help
with Lemma~\ref{prop:separate leafs analysis}, and to Alfons W. Laarman,
Tom van Dijk and the anonymous peer reviewers for their valuable feedback.
Finally, thanks to the Centre for Scientific Computing Aarhus,
(\href{http://phys.au.dk/forskning/cscaa/}{phys.au.dk/forskning/cscaa/}) to
allow us to run our benchmarks on their Grendel cluster.

\bibliographystyle{plain}
\bibliography{references}

@Article{Aggarwal1987,
  title       = {The Input/Output Complexity of Sorting and Related Problems},
  author      = {Aggarwal, Alok
             and Vitter, Jeffrey,S.},
  year        = 1988,
  journal     = {Communications of the ACM},
  publisher   = {Association for Computing Machinery},
  volume      = 31,
  number      = 9,
  doi         = {10.1145/48529.48535},
  pages       = {1116--1127},
}

@InProceedings{Amaru2015,
  author    = {Amar{\'u}, Luca
           and Gaillardon, Pierre-Emmanuel
           and De Micheli, Giovanni},
  title     = {The {EPFL} combinational benchmark suite},
  booktitle = {24th International Workshop on Logic \& Synthesis},
  year      = 2015,
}

@InProceedings{Amparore2020,
  title     = {A {CTL}* Model Checker for {P}etri Nets},
  author    = {Amparore, Elvio
           and Donatelli, Susanna
           and Gallà, Francesco},
  year      = 2020,
  booktitle = {Application and Theory of {P}etri Nets and Concurrency},
  publisher = {Springer},
  series    = {Lecture Notes in Computer Science},
  volume    = 12152,
  pages     = {403--413},
  doi       = {10.1007/978-3-030-51831-8\_21},
}

@InProceedings{Arge1995:1,
  title     = {The {I/O}-complexity of ordered binary-decision diagram manipulation},
  author    = {Arge, Lars},
  booktitle = {6th International Symposium on Algorithms and Computations (ISAAC)},
  year      = 1995,
  series    = {Lecture Notes in Computer Science},
  volume    = 1004,
  pages     = {82--91},
  doi       = {10.1007/BFb0015411},
}

@InProceedings{Arge1995:2,
  author    = {Arge, Lars},
  title     = {The buffer tree: A new technique for optimal {I/O}-algorithms},
  booktitle = {Workshop on Algorithms and Data Structures (WADS)},
  year      = 1995,
  series    = {Lecture Notes in Computer Science},
  volume    = 955,
  publisher = {Springer},
  address   = {Berlin, Heidelberg},
  pages     = {334--345},
  doi       = {10.1007/3-540-60220-8\_74},
}

@InProceedings{Arge1996,
    author    = {Arge, Lars},
    title     = {The {I/O}-Complexity of Ordered Binary-Decision Diagram},
    booktitle = {BRICS RS preprint series},
    volume    = 29,
    year      = 1996,
    publisher = {Department of Computer Science, University of Aarhus},
    doi       = {10.7146/brics.v3i29.20010},
}

@InProceedings{Ashar1994,
  author    = {Ashar, Pranav
           and Cheong, Matthew},
  title     = {Efficient Breadth-First Manipulation of Binary Decision Diagrams},
  year      = 1994,
  booktitle = {IEEE/ACM International Conference on Computer-Aided Design (ICCAD)},
  publisher = {IEEE Computer Society Press},
  pages     = {622--627},
  doi       = {10.1109/ICCAD.1994.629886}
}

@InProceedings{BenDavid2000,
  author    = {Ben-David, Shoham
           and Heyman, Tamir
           and Grumberg, Orna
           and Schuster, Assaf},
  title     = {Scalable Distributed On-the-Fly Symbolic Model Checking},
  booktitle = {Formal Methods in Computer-Aided Design},
  year      = 2000,
  publisher = {Springer},
  address   = {Berlin, Heidelberg},
  pages     = {427--441},
  doi       = {10.1007/3-540-40922-X\_24}
}

@InProceedings{Brace1990,
  title     = {Efficient implementation of a {BDD} package},
  author    = {Brace, Karl S.  and Rudell, Richard L. and Bryant, Randal E.},
  booktitle = {27th Design Automation Conference (DAC)},
  year      = {1990},
  publisher = {Association for Computing Machinery},
  pages     = {40--45},
  doi       = {10.1109/DAC.1990.114826},
}

@Article{Bryant1986,
  title     = {Graph-Based Algorithms for {B}oolean Function Manipulation},
  author    = {Bryant, Randal E.},
  year      = 1986,
  journal   = {IEEE Transactions on Computers},
  volume    = {C-35},
  number    = 8,
  pages     = {677--691},
  publisher = {IEEE Computer Society Press},
  doi       = {10.1109/TC.1986.1676819},
}

@Article{Bryant1992,
  title     = {Symbolic {B}oolean Manipulation with Ordered Binary-Decision Diagrams},
  author    = {Bryant, Randal E.},
  year      = 1992,
  publisher = {Association for Computing Machinery},
  journal   = {ACM Computing Surveys},
  volume    = 24,
  number    = 3,
  pages     = {293--318},
  doi       = {10.1145/136035.136043},
}

@InCollection{Bryant2018,
  author    = {Bryant, Randal E.},
  editor    = {Clarke, Edmund M.
           and Henzinger, Thomas A.
           and Veith, Helmut
           and Bloem, Roderick},
  title     = {Binary Decision Diagrams},
  booktitle = {Handbook of Model Checking},
  year      = 2018,
  publisher = {Springer International Publishing},
  address   = {Cham},
  pages     = {191--217},
  isbn      = {978-3-319-10575-8},
  doi       = {10.1007/978-3-319-10575-8}
}

@InProceedings{Bryant2020,
  author    = {Bryant, Randal E.
           and Heule, Marijn J. H.},
  title     = {Generating Extended Resolution Proofs with a {BDD}-Based SAT Solver},
  booktitle = {Tools and Algorithms for the Construction and Analysis of Systems (TACAS)},
  year      = 2021,
  publisher = {Springer International Publishing},
  series    = {Lecture Notes in Computer Science},
  volume    = {12651},
  pages     = {76--93},
  doi       = {10.1007/978-3-030-72016-2\_5}
}

@Misc{Bryant2021:1,
  title         = {Generating Extended Resolution Proofs with a {BDD}-Based SAT Solver},
  author        = {Randal E. Bryant and Marijn J. H. Heule},
  year          = 2021,
  howpublished  = {arXiv},
  primaryClass  = {cs.LO},
  url           = {https://arxiv.org/abs/2105.00885},
}

@InProceedings{Bryant2021:2,
  author    = {Bryant, Randal E.
           and Heule, Marijn J. H.},
  title     = {Dual Proof Generation for Quantified {B}oolean Formulas with a {BDD}-based Solver},
  booktitle = {28th International Conference on Automated Deduction (CADE)},
  year      = 2021,
  publisher = {Springer International Publishing},
  series    = {Lecture Notes in Computer Science},
  volume    = {12699},
  pages     = {433--449},
  isbn      = {978-3-030-79876-5},
  doi       = {10.1007/978-3-030-79876-5\_25},
}

@Misc{Chen2016,
  title         = {Robust benchmarking in noisy environments},
  author        = {Jiahao Chen and Jarrett Revels},
  year          = 2016,
  howpublished  = {arXiv},
  primaryClass  = {cs.PF},
  url           = {https://arxiv.org/abs/1608.04295},
}

@Article{Cimatti2000,
  title   = {{N}U{SMV}: a new symbolic model checker},
  author  = {Cimatti, Alessandro
         and Clarke, Edmund
         and Giunchiglia, Fausto
         and Roveri, Marco},
  year    = 2000,
  journal = {International Journal on Software Tools for Technology Transfer},
  volume  = 2,
  pages   = {410--425},
  doi     = {10.1007/s100090050046},
}

@InProceedings{Cimatti2002,
  author    = {Cimatti, Alessandro
           and Clarke, Edmund
           and Giunchiglia, Enrico
           and Giunchiglia, Fausto
           and Pistore, Marco
           and Roveri, Marco
           and Sebastiani, Roberto
           and Tacchella, Armando},
  title     = {{N}u{SMV} 2: An OpenSource Tool for Symbolic Model Checking},
  booktitle = {International Conference on Computer Aided Verification (CAV)},
  year      = {2002},
  publisher = {Springer},
  address   = {Berlin, Heidelberg},
  series    = {Lecture Notes in Computer Science},
  volume    = {2404},
  pages     = {359--364},
  isbn      = {978-3-540-45657-5},
  doi       = {10.1007/3-540-45657-0\_29},
}

@Article{Dijk2016,
  title   = {Sylvan: multi-core framework for decision diagrams},
  author  = {{\noopsort{Dijk}{Van Dijk}}, Tom
         and {\noopsort{Pol}{Van de Pol}}, Jaco},
  year    = 2016,
  journal = {International Journal on Software Tools for Technology Transfer},
  volume  = 19,
  pages   = {675--696},
  doi     = {10.1007/s10009-016-0433-2}
}

@InProceedings{Frigo1999,
  title     = {Cache-Oblivious Algorithms},
  author = {Frigo, Matteo
        and Leiserson, Charles E.
        and Prokop, Harald
        and Ramachandran, Sridhar},
  year      = 1999,
  publisher = {IEEE Computer Society Press},
  booktitle = {40th Annual Symposium on Foundations of Computer Science (FOCS 1999)},
  pages     = {285--297},
  doi       = {10.1109/SFFCS.1999.814600},
}

@Article{Fujita1997,
  title   = {Multi-Terminal Binary Decision Diagrams: An Efficient Data Structure for Matrix Representation},
  author  = {Fujita, M.
        and McGeer, P.C.
        and Yang, J.C.-Y.},
  year    = 1997,
  journal = {Formal Methods in System Design},
  volume  = 10,
  pages   = {149--169},
  doi     = {10.1023/A:1008647823331},
}

@InProceedings{Gammie2004,
  author    = {Gammie, Peter
           and {\noopsort{Meyden}{Van der Meyden}}, Ron},
  title     = {{MCK}: Model Checking the Logic of Knowledge},
  booktitle = {Computer Aided Verification},
  series    = {Lecture Notes in Computer Science},
  volume    = 3114,
  year      = 2004,
  publisher = {Springer},
  address   = {Berlin, Heidelberg},
  pages     = {479--483},
  doi       = {10.1007/978-3-540-27813-9\_41}
}

@Article{Grumberg2005,
  author  = {Grumberg, Orna
         and Heyman, Tamir
         and Schuster, Assaf},
  title   = {Distributed Symbolic Model Checking for $\mu$-Calculus},
  year    = 2005,
  journal = {Formal Methods in System Design},
  volume  = 26,
  pages   = {197--219},
  doi     = {10.1007/s10703-005-1493-1},
}

@Misc{He2020,
  title        = {{P}etri Net Based Symbolic Model Checking for Computation Tree Logic of Knowledge},
  author       = {Leifeng He
              and Guanjun Liu},
  year         = 2020,
  primaryClass = {cs.SE},
  howpublished = {arXiv},
  url          = {https://arxiv.org/abs/2012.10126},
}

@Book{Knuth2011,
  title     = {The art of computer programming, volume 4A: combinatorial algorithms, part 1},
  author    = {Knuth, Donald E},
  year      = 2011,
  publisher = {Addison-Wesley Professional},
}

@InProceedings{Kant2015,
  author    = {Kant, Gijs
          and Laarman, Alfons
          and Meijer, Jeroenn
          and {\noopsort{Pol}{Van de Pol}}, Jaco
          and Blom, Stefan
          and {\noopsort{Dijk}{Van Dijk}}, Tom},
  title     = {{LTS}min: High-Performance Language-Independent Model Checking},
  booktitle = {Tools and Algorithms for the Construction and Analysis of Systems (TACAS)},
  year      = {2015},
  publisher = {Springer},
  address   = {Berlin, Heidelberg},
  series    = {Lecture Notes in Computer Science},
  volume    = {9035},
  pages     = {692--707},
  doi       = {10.1007/978-3-662-46681-0\_61},
}

@TechReport{Karplus1988,
  author      = {Karplus, Kevin},
  title       = {Representing {B}oolean Functions with If-Then-Else {DAG}s},
  year        = 1988,
  institution = {University of California at Santa Cruz},
  address     = {USA}
}

@InProceedings{Klarlund1996,
  title     = {{BDD} Algorithms and Cache Misses},
  author    = {Klarlund, Nils and Rauhe, Theis},
  year      = {1996},
  booktitle = {BRICS Report Series},
  volume    = 26,
  doi       = {10.7146/brics.v3i26.20007},
}

@InProceedings{Kunkle2010,
  title     = {Parallel Disk-Based Computation for Large, Monolithic Binary Decision Diagrams},
  author    = {Kunkle, Daniel
           and Slavici, Vlad
           and Cooperman, Gene},
  year      = 2010,
  booktitle = {4th International Workshop on Parallel Symbolic Computation (PASCO)},
  pages     = {63--72},
  doi       = {10.1145/1837210.1837222},
}

@TechReport{Lind1999,
  title       = {{B}u{DD}y: A binary decision diagram package},
  author      = {Lind-Nielsen, J{\o}rn},
  year        = {1999},
  institution = {Department of Information Technology, Technical University of Denmark},
}

@Article{Lomuscio2017,
  author  = {Lomuscio, Alessio
         and Qu, Hongyang
         and Raimondi, Franco},
  title   = {{MCMAS}: an open-source model checker for the verification of multi-agent systems},
  journal = {International Journal on Software Tools for Technology Transfer},
  volume  = 19,
  pages   = {9--30},
  year    = 2017,
  doi     = {10.1007/s10009-015-0378-x},
}

@InProceedings{Long1998,
  title     = {The Design of a Cache-Friendly {BDD} Library},
  author    = {Long, David E.},
  year      = 1998,
  booktitle = {Proceedings of the 1998 IEEE/ACM International Conference on Computer-Aided Design (ICCAD)},
  publisher = {Association for Computing Machinery},
  pages     = {639--645},
}

@InProceedings{Minato1990,
  title     = {Shared binary decision diagram with attributed edges for efficient {B}oolean function manipulation},
  author    = {Minato, Shin-ichi
           and Ishiura, Nagisa
           and Yajima, Shuzo},
  booktitle = {27th Design Automation Conference (DAC)},
  year      = 1990,
  pages     = {52--57},
  publisher = {Association for Computing Machinery},
  doi       = {10.1145/123186.123225}
}

@InProceedings{Minato1993,
  title     = {Zero-Suppressed {BDD}s for Set Manipulation in Combinatorial Problems},
  author    = {Minato, Shin-ichi},
  year      = 1993,
  isbn      = {0897915771},
  booktitle = {30th Design Automation Conference (DAC)},
  pages     = {272--277},
  publisher = {Association for Computing Machinery},
  doi       = {10.1145/157485.164890}
}

@InProceedings{Minato2001,
  author    = {Minato, Shin-ichi
           and Ishihara, Shinya},
  title     = {Streaming {BDD} Manipulation for Large-Scale Combinatorial Problems},
  year      = {2001},
  booktitle = {Design, Automation and Test in Europe Conference and Exhibition},
  pages     = {702--707},
  doi       = {10.1109/DATE.2001.915104}}

@Article{Molhave2012,
  author    = {M{\o}lhave, Thomas},
  title     = {Using {TPIE} for Processing Massive Data Sets in {C}++},
  year      = 2012,
  publisher = {Association for Computing Machinery},
  journal   = {SIGSPATIAL Special},
  volume    = 4,
  number    = 2,
  pages     = {24--27},
  doi       = {10.1145/2367574.2367579},
}

@InProceedings{Ochi1993,
  title     = {Breadth-first manipulation of very large binary-decision diagrams},
  author    = {Ochi, Hiroyuki
           and Yasuoka, Koichi
           and Yajima, Shuzo},
  booktitle = {International Conference on Computer Aided Design (ICCAD)},
  publisher = {IEEE Computer Society Press},
  year      = 1993,
  pages     = {48--55},
  doi       = {10.1109/ICCAD.1993.580030},
}

@MastersThesis{Petersen2007,
  title   = {External Priority Queues in Practice},
  author  = {Petersen, Lars Hvam},
  school  = {Department of Computer Science, University of Aarhus},
  year    = 2007
}

@Article{Sanders2001,
  author    = {Sanders, Peter},
  title     = {Fast Priority Queues for Cached Memory},
  year      = 2000,
  publisher = {Association for Computing Machinery},
  journal   = {ACM Journal of Experimental Algorithmics},
  volume    = 5,
  pages     = {7--32},
  doi       = {10.1145/351827.384249},
}

@InProceedings{Sanghavi1996,
  author    = {Sanghavi, Jagesh V.
           and Ranjan, Rajeev K.
           and Brayton, Robert K.
           and Sangiovanni-Vincentelli, Alberto},
  title     = {High Performance {BDD} Package by Exploiting Memory Hierarchy},
  year      = 1996,
  booktitle = {33rd Design Automation Conference (DAC)},
  pages     = {635--640},
  publisher = {Association for Computing Machinery},
  doi       = {10.1145/240518.240638}
}

@TechReport{Somenzi2015,
  title       = {{CUDD}: {CU} Decision Diagram Package, 3.0},
  author      = {Somenzi, Fabio},
  year        = 2015,
  institution = {University of Colorado at Boulder},
}

@Misc{Soelvsten2021:Artifact:arXiv,
  author       = {S{\o}lvsten, Steffan Christ
              and {\noopsort{Pol}{Van de Pol}}, Jaco},
  howpublished = {Zenodo},
  title        = {Adiar v1.0.1 : Experiment Data},
  year         = 2021,
  doi          = {10.5281/zenodo.5638551},
}

@InProceedings{Soelvsten2022:TACAS,
  title         = {Adiar: Binary Decision Diagrams in External Memory},
  author        = {S{\o}lvsten, Steffan Christ
               and van de Pol, Jaco
               and Jakobsen, Anna Blume
               and Thomasen, Mathias Weller Berg},
  year          = {2022},
  booktitle     = {Tools and Algorithms for the Construction and Analysis of Systems},
  pages         = {295--313},
  numPages      = {19},
  publisher     = {Springer},
  series        = {Lecture Notes in Computer Science},
  volume        = {13244},
  doi           = {10.1007/978-3-030-99527-0\_16},
}

@InProceedings{Vengroff1994,
  author    = {Vengroff, Darren Erik},
  title     = {A {T}ransparent {P}arallel {I/O} {E}nvironment},
  booktitle = {DAGS Symposium on Parallel Computation},
  year      = 1994,
  pages     = {117--134}
}

@InProceedings{Velev2014,
  title     = {Efficient parallel {GPU} algorithms for {BDD} manipulation},
  author    = {Velev, Miroslav N.
           and Gao, Ping},
  booktitle = {19th Asia and South Pacific Design Automation Conference (ASP-DAC)},
  year      = 2014,
  pages     = {750--755},
  publisher = {IEEE Computer Society Press},
  doi       = {10.1109/ASPDAC.2014.6742980},
}

\end{document}